\documentclass[12pt,journal,onecolumn]{IEEEtran}
\usepackage{amsthm}
\usepackage{arydshln}
\usepackage{amssymb}
\usepackage{amsmath}
\usepackage{longtable}
\newtheorem{theorem}{Theorem}[section]
\newtheorem{lemma}[theorem]{Lemma}

\newtheorem{corollary}[theorem]{Corollary}
\newtheorem{example}{Example}

\allowdisplaybreaks
\usepackage[justification=centering]{caption}

\usepackage[para]{threeparttable}

\newcommand{\tr}{{\mathrm{Tr}}}

\newcommand{\gf}{{\mathbb{F}}}

\newcommand{\opt}{{\text{opt}}}

\usepackage{lineno}
\usepackage{blindtext}
\usepackage{threeparttable}
 \usepackage{graphicx}
 \usepackage{arydshln}
\allowdisplaybreaks[4]

\makeatletter

\newcommand{\Rmnum}[1]{\expandafter\@slowromancap\romannumeral #1@}
\makeatother

\ifCLASSINFOpdf

\else

\fi

\begin{document}
%
\title{Large Sets of Quasi-Complementary Sequences\\ From Polynomials over Finite Fields and\\ Gaussian Sums
\thanks{
Z. Heng's research was supported in part by the National Natural Science Foundation of China under Grant 12271059, in part by the Shaanxi Fundamental
 Science Research Project for Mathematics and Physics (Grant No. 23JSZ008), in part by the open research fund of National Mobile Communications 
 Research Laboratory of Southeast University under Grant 2024D10 and in part by the Research Funds for the Interdisciplinary Projects, CHU, under Grant 300104240922. C. Xie's research was supported by the Shaanxi Provincial Key Research and Development Program under Grant 2023-YBGY-042. H. Zhou's research was supported by the National Natural Science Foundation of China under Grant 12471493.
}}

\author{Ziling Heng, Peng Wang, Chunlei Xie and Haiyan Zhou \thanks{Z. Heng and P. Wang are with the School of Science, Chang'an University, Xi'an 710064, China,  and also with National Mobile Communications Research Laboratory, Southeast University, Nanjing 211111, China (email: zilingheng@chd.edu.cn, wp20201115@163.com); C. Xie is with the School of Information Engineering, Chang'an University, Xi'an 710064, China (email: chunleixie@chd.edu.cn); H. Zhou is with the School of Mathematical Sciences and Institute of Mathematics, Nanjing Normal University, Nanjing 210023, China (email: 05366@njnu.edu.cn). }}

\date{\today}
\maketitle

\begin{abstract}
Perfect complementary sequence sets (PCSSs)  are widely used in multi-carrier code-division multiple-access (MC-CDMA) communication systems.  However, the set size of a PCSS is upper bounded by the number of row sequences of each two-dimensional matrix in the PCSS.  Then quasi-complementary sequence sets (QCSSs) were proposed to support more users in MC-CDMA communications. For practical applications, it is desirable to construct an $(M,K,N,\vartheta_{\max})$-QCSS with $M$ as large as possible and $\vartheta_{max}$ as small as possible, 
where $M$ is the number of matrices with $K$ rows and $N$ columns in the set and $\vartheta_{\max}$ denotes its periodic tolerance.
There exists a tradeoff among these parameters. Constructing QCSSs achieving or nearly achieving the  known correlation lower bound has been an interesting research topic. 
Up to now, only a few constructions of asymptotically optimal or near-optimal periodic QCSSs have been reported in the literature. In this paper, based on polynomials over finite fields and Gaussian sums, we construct five new families of  asymptotically optimal or near-optimal periodic QCSSs with large set sizes and low periodic tolerances. These families of QCSSs have set size $\Theta(K^2)$ or $\Theta(K^3)$ and flock size $K$. To the best of our knowledge, only a small amount of known families of periodic QCSSs with set size $\Theta(K^2)$ have been constructed and most of other known periodic QCSSs have set sizes much smaller than $\Theta(K^2)$.
Our new constructed periodic QCSSs with set size $\Theta(K^2)$ and flock size $K$ have the best parameters among all known ones. They have larger set sizes or lower periodic tolerances.
The periodic QCSSs with set size $\Theta(K^3)$ and flock size $K$ constructed in this paper have the largest set size among all known families of asymptotically optimal or near-optimal periodic QCSSs.
\end{abstract}

\begin{IEEEkeywords}
Quasi-complementary sequence sets, polynomials over finite fields, Gaussian sums
\end{IEEEkeywords}

\maketitle{}

%
\IEEEpeerreviewmaketitle

\section{Introduction}\label{sec1}
A perfect complementary sequence set (PCSS), which is also called mutually orthogonal complementary sequence set, consists of a number of two-dimensional matrices. The nontrivial auto-correlation and cross-correlation sums of the row sequences of the matrices are zero for any nonzero time-shift \cite{TL}.  
PCSSs have found many nice applications including peak-to-average power ration reduction \cite{D}, inter-symbol interference channel estimation \cite{S}, radar waveform design \cite{P} and  multi-carrier code-division multiple-access (MC-CDMA) communication systems \cite{C1}. The application of PCSSs has been viewed as a very promising technology for the next generation CDMA \cite{C2}. In an MC-CDMA communication system, each user is assigned a two-dimensional matrix and all of its row sequences are transmitted over separate sub-carrier channels simultaneously \cite{C1}.  Then the number of the row sequences of a matrix in the PCSS and that of sub-carrier channels are equal. 
However, the set size of a PCSS can not be larger than the number of row sequences of each two-dimensional matrix in it by the correlation lower bound.
In order to support more users in MC-CDMA communication systems, some researchers proposed two new kinds of complementary sequence sets which are zero-correlation zone complementary sequence sets (ZCZ-CSSs) \cite{F, L} and  low-correlation zone complementary sequence sets (LCZ-CSS) \cite{Liu}. An LCZ-CSS (resp. ZCZ-CSS) has low (resp. zero) correlation magnitudes within a zone around the zero-shift. In recent years, many researches have been done on constructing LCZ-CSSs and ZCZ-CSSs.

In 2013, Liu et al. proposed another type of complementary sequence sets, i.e. the quasi-complementary sequence sets (QCSSs),  which could have larger set sizes than those of PCSSs \cite{ZL}.
In the literature, there are two kinds of QCSSs including periodic QCSSs and aperiodic QCSSs \cite{ZL, Liu2}. In this paper, we mainly study periodic QCSSs. 
A periodic $(M, K, N, \vartheta_{\max})$-QCSS over an alphabet $A$ is a set consisting of $M$ two-dimensional $K\times N$ matrices 
 such that the maximum value $\vartheta_{\max}$ of nontrivial periodic correlation sums of the row sequences is a small positive number.   
 Each row of the two-dimensional matrices is called a constituent sequence of length $N$ and $K$ is called the flock size. 
The maximum non-trivial periodic correlation sum $\vartheta_{\max}$ measures the maximum possible amount of the multipath interference and the multiuser interference. 
For the application of a periodic QCSSs in MC-CDMA communications, it is required that its set size $M$ is as large as possible and $\vartheta_{\max}$ is as small as possible.
However, there exists a tradeoff among the parameters $M, K, N$ and $\vartheta_{\max}$. 
In \cite{ZL}, Liu et al. presented a lower bound for the maximum periodic correlation
sum of a periodic $(M, K ,N, \vartheta_{\max})$-QCSS as follows:
\begin{eqnarray}\label{wlech}
\vartheta_{\max}\geq\vartheta_{\opt}=KN\sqrt{\frac{M/K-1}{MN-1}}.
\end{eqnarray}
It is easy to verify that $M\leq K$, which is the upper bound of the size of a periodic PCSS, provided that $\vartheta_{\max}=0$. 
The tightness factor of the above correlation lower bound for a periodic QCSS is defined by $\rho=\vartheta_{\max}/\vartheta_{\opt}$. 
Note that $\rho\geq1$. A QCSS is said to be \emph{optimal} if $\rho=1$ and\emph{ near-optimal} if $1<\rho\leq2$ \cite{ZL}.
Constructing optimal or near-optimal periodic QCSSs has been an interesting research topic in recent years. However, up to now, no infinite family of optimal QCSSs achieving the lower bound in (\ref{wlech}) was reported. Only a few families of asymptotically optimal or near-optimal QCSSs have been constructed. 
In \cite{ZL}, Liu et al. constructed the first family of asymptotically optimal and the first family of asymptotically near-optimal periodic QCSSs from the Single difference sets and the existing optimal quaternary sequence sets. Using a similar idea, Li et al. constructed a family of asymptotically optimal periodic QCSSs from almost difference sets in \cite{LY4}. 
In \cite{LY1}, Li et al. generalized the constructions in \cite{ZL} and derived new asymptotically optimal periodic QCSSs. 
In \cite{LY2, LY3}, Li et al. presented several constructions of asymptotically optimal QCSSs from additive and multiplicative characters over finite fields. 
Later, Luo et al. gave three new constructions of asymptotically optimal QCSSs with small alphabet sizes by additive characters over finite fields \cite{LG}. 
Recently, Xiao et al. used additive characters of finite fields to construct two more families of asymptotically optimal periodic QCSSs \cite{XLC}.
The parameters of these known families of periodic QCSSs are summarized in Table \ref{tab1}.  It is shown that only a small amount of known families of periodic QCSSs with set size $\Theta(K^2)$ and flock size $K$ have been constructed and most of other known periodic QCSSs have set sizes much smaller than $\Theta(K^2)$.
The purpose of this paper is to use new approaches to construct five new families of  asymptotically optimal or near optimal QCSSs with better parameters.
The main contents of this paper are as follows:
\begin{enumerate}
\item[$\bullet$] Using a family of quadratic polynomials over finite fields, we construct a family of $(q^2,q,q-1,q)$-QCSSs with alphabet size $p$ and a family of $(2^{2n},2^n,2^{n}-1,2^n)$-QCSSs with alphabet size $2$, where $q=p^n$ with $p$ an odd prime and $n$ a positive integer. These two families of QCSSs are asymptotically optimal  with respect to the correlation lower bound and have larger set sizes than those in \cite{LY2, XLC} by Table \ref{tab1}. 
\item[$\bullet$] A family of $(2^{3n},2^n,2^n-1,2^{n+1})$-QCSSs with alphabet size $2$ is constructed from a special quadrinomial over finite fields with characteristic 2. 
A family of $(p^{3n},p^n,p^n-1, 2p^n)$-QCSSs with alphabet size $p$ is also constructed from a cubic quadrinomial over finite fields with characteristic $p$ for an odd prime $p$. 
These two families of QCSSs are asymptotically near-optimal with respect to the correlation lower bound. Note that the set size of each of them is the cube of the flock size. In this sense, these two families of QCSSs have the largest set sizes among all known families of asymptotically optimal or near-optimal periodic QCSSs.
\item[$\bullet$] By Gaussian sums, we construct a family of  $(p^n(p^n-1),p^n-1,p^n-1,p^n)$-QCSSs which have  lower periodic tolerance than the known  $(p^n(p^n-1),p^n-1,p^n-1,p^n+1)$-QCSSs
in \cite{XLC}, though our QCSSs have larger alphabet size. 
\end{enumerate}
According to the parameters of all known periodic QCSSs listed in Table \ref{tab1}, our QCSSs have new parameters and very large set sizes. Most of them have very small alphabet size.
These new QCSSs may be used in MC-CDMA communication systems to support a large number of users. 
\begin{table}[h!]
\begin{center}
\scriptsize{\begin{threeparttable}
\caption{The parameters of known periodic QCSSs.}\label{tab1}
\begin{tabular}{l|l|l|l|l|l|l}
\hline 
Set size $M$& Flock size $K$& Length $N$& $\vartheta_{\max}$ & Alphabet  Size& Constraint& References\\
\hline 
$2^n$ & $2^{n-1}-1$ & $2^n-1$& $\frac{2^n+2^\frac{n}{2}}{2}$ & $4(2^n-1)$&$n>1$ &\cite{ZL} \\
\hline 
$2^n$&$2^{n-1}-1$&$2^{n+1}-2$&$2^n+2^{n/2}$&$4(2^n-1)$& $n>1$&\cite{ZL} \\
\hline 
$2^{2n}$&$2^{n-1}-2^{n/2}$&$2^{2n}-1$&$\leq w$&$4(2^n-1)$&$n\geq 3$&\cite{LY4}\\
\hline 
$p$ & $\frac{p-1}{2}$ &$p$ & $\leq \frac{p+\sqrt{p}}{2}$& $p$&$p\equiv1 \pmod{4}$ is a prime &\cite{LY1} \\
\hline 
$p^n-2$ & $\frac{p^n-1}{2}$ & $p^n-1$ &$\leq \frac{p^n+4 \sqrt{p^n}+3}{2}$ &$p^n-1$ & $p$ is an odd prime&\cite{LY2}\\
\hline 
$p^{2n}-2$ & $p^n$&$p^{2n}-1$ &$p^n(p^{2n}+3)$ & $p^{2n}-1$&$p$ is a prime&\cite{LY2}\\
\hline 
$p^n-1$ &   $\frac{p^n-1}{2}$&$p^n-1$ & $\leq \frac{p^n+\sqrt{p^n}}{2}$&$p(p^n-1)$& $p$ is an odd prime&\cite{LY3}\\
 \hline 
$p^n-1$ & $p^{n-1}$&$p^n-1$ &$\leq p^{n-\frac{1}{2}}$ & $p(p^{n}-1)$&$p$ is a prime and $n>1$ &\cite{LY3}\\
\hline 
$p^{2n}-1$ &$p^n$ & $p^{2n}-1$& $p^{3n/2}$ &$p(p^{2n}-1)$ & $p$ is a prime&\cite{LY3}\\
\hline 
$2^n-1$& $2^{n-1}-1$& $2^n-1$&$2^{n-1}$ & $2(2^n-1)$&$n>1$ &\cite{LY3}\\
\hline 
$2^n-1$&$2^{n-2}$ &$2^n-1$ &$\leq3\cdot2^{n-2}$ &$2(2^n-1)$&$n>1$&\cite{LY3}\\
\hline 
$p^n$&$\frac{p^n-1}{2} $&$p^n-1$&$\frac{p^n+1}{2}$&$p$& $p$ is an odd prime, $p^n>3$,&\cite{LG}\\
\hline 
$p^n$&$\frac{p^n-p^{n-1}}{2}$&$p^n-1$&$\frac{p^n+p^{n-1}}{2}$&$p$&$p$ is an odd prime, $p^n>3$,&\cite{LG}\\
\hline 
$p^n$&$u$&$p^n-1$&$v$&$p$& $p$ is an  odd prime, $n>1$ is odd, &\cite{LG}\\
\hline 
$p^{2n}-p^n$&$p^n$&$p^n-1$&$p^n$&$p$& $p$ is a prime, $n>1$&\cite{XLC}\\
\hline 
$p^{2n}-p^n$&$p^n-1$&$p^n-1$&$p^n+1$&$p$& $p$ is a prime, $n>1$&\cite{XLC}\\
\hline 
$p^{2n}$&$p^n$&$p^n-1$&$p^n$&$p$&$p$ is an odd prime, $n\geq 1$&Theorem \ref{111}\\
\hline 
$2^{2n}$&$2^n$&$2^n-1$&$2^n$&$2$&$n\geq 1$&Theorem \ref{222}\\
\hline 
$2^{3n}$&$2^n$&$2^n-1$&$2^{n+1}$&$2$&$n\geq 3$&Theorem \ref{444}\\
\hline 
$p^{3n}$&$p^n$&$p^n-1$&$ 2p^n$&$p$&$p$ is an odd prime, $n\geq 1$,&Theorem \ref{333}\\
\hline 
$p^{2n}-p^n$&$p^n-1$&$p^n-1$&$p^n$&$p(p^n-1)$&$p$ is a prime&Theorem \ref{555}\\
\hline 
\end{tabular}
\begin{tablenotes}
\footnotesize
\item[]$u=\frac{p^n-p^{n-1}+(-1)^{(p-1)(n+3)/4}(p-1)p^{(n-1)/2}}{2}$, $v=\frac{p^n+p^{n-1}-(-1)^{(p-1)(n+3)/4}(p-1)p^{(n-1)/2}}{2}$,\\$w=(1+2^n)\sqrt{2^n+2^{n-2}-2^{n/2}}$
\end{tablenotes}
\end{threeparttable}}
\end{center}
\end{table}

\section{Preliminaries}\label{sec2}
In this section, we recall some basic definitions of QCSSs and results on characters over finite fields, Gaussian sums, character sums, and the numbers of zeros of some polynomials over finite fields.
\subsection{Periodic quasi-complementary sequence set}
For two complex-valued sequences $\textbf{a}=(a_0,a_1,\cdots,a_{N-1})$ and $\textbf{b}=(b_0,b_1,\cdots,b_{N-1})$ of period $N$, the periodic correlation function between them is defined by
$$R_{\textbf{a},\textbf{b}}(\tau)=\sum_{t=0}^{N-1}a_t\overline{b_{t+\tau}}$$
for $0\leq \tau \leq N-1$, where $(t+\tau)$ is calculated modulo $N$ and $\overline{b_{t+\tau}}$ is the complex conjugation of $b_{t+\tau}$.

Let $\mathbb{C}=\left\{\mathbf{C}^{0},\mathbf{C}^{1},\cdots,\mathbf{C}^{M-1}\right\}$ be a set of $M$ complementary sequences such that each $\mathbf{C}^{m}$ is a  two-dimensional matrix with size $K \times N$ given by
\begin{eqnarray*}
\mathbf{C}^{m}=\left[
\begin{array}{cccc}
\mathbf{c}_{0}^{m}\\
\mathbf{c}_{1}^{m}\\
\vdots\\
\mathbf{c}_{K-1}^{m}\\
\end{array}\right]
=
\left[
\begin{array}{cccc}
c_{0,0}^{m} & c_{0,1}^{m}  & \cdots & c_{0,N-1}^{m}\\
c_{1,0}^{m}  & c_{1,1}^{m}  & \cdots & c_{1,N-1}^{m}\\
\vdots   & \vdots  & \ddots & \vdots \\
c_{K-1,0}^{m}  & c_{K-1,1}^{m} & \cdots & c_{K-1,N-1}^{m} \\
\end{array}\right],
\end{eqnarray*}
where $\mathbf{c}_{k}^{m}=(c_{k,0}^{m}, c_{k,1}^{m},\cdots, c_{k,N-1}^{m})$ is the $k$-th constituent sequence of length $N$, $0\leq k \leq K-1$, $0\leq m\leq M-1$.
The periodic correlation function between two complementary sequences $\mathbf{C}^{{m}_{1}}$ and $\mathbf{C}^{{m}_{2}}$ is defined by 
\begin{eqnarray*}
R_{{{\textbf{C}}^{{m}_{1}}},{{\textbf{C}}^{{m}_{2}}}}(\tau)=\sum_{k=0}^{K-1}R_{{\mathbf{c}_{k}^{{m}_{1}}},{\mathbf{c}_{k}^{{m}_{2}}}}(\tau), ~ 0\leq \tau < N ,
\end{eqnarray*}
where $0\leq m_{1},m_{2}\leq M-1$. 
The  maximum periodic auto-correlation magnitude and the maximum periodic cross-correlation magnitude of  $\mathbb{C}$ are  respectively defined by
\begin{eqnarray*}
\vartheta_{a} = \max\left\{|R_{{{\textbf{C}}^{{m}}},{{\textbf{C}}^{{m}}}}(\tau)|:0\leq m< M,\ 0<\tau< N\right\},\\
\vartheta_{c} = \max\left\{|R_{{{\textbf{C}}^{{m}_{1}}},{{\textbf{C}}^{{m}_{2}}}}(\tau)|:0\leq m_{1}\neq m_{2}< M,\ 0\leq\tau< N\right\}.
\end{eqnarray*}
Define the maximum periodic correlation magnitude (also called periodic tolerance) by
\begin{eqnarray*}
 \vartheta_{\max} = \max\{\vartheta_{c}, \vartheta_{a}\}.
\end{eqnarray*}
It is an important performance measure of the complementary sequence
set $\mathbb{C}$ in practice applications.
If $\vartheta_{\max}> 0$, then $\mathbb{C}$ is called a periodic $(M, K, N, \vartheta_{\max})$ quasi-complementary sequence set (QCSS). Otherwise, $\mathbb{C}$ is called a periodic perfect complementary sequence set (PCSS). 
Particularly, a PCSS reduces to a matrix consisting of two row sequences provided that $M = 1$ and $K = 2$. This matrix is called a Golay complementary pair.

\subsection{Additive and multiplicative characters of finite fields}\label{subsection}
For a prime $p$  and  a positive integer $n$, let $q=p^n$ and $\zeta_p$ be the primitive $p$-th root of complex unity. 
Let $\gf_q$ denote the finite field with $q$ elements. Define the trace function from $\gf_q$ onto $\gf_p$ by
$$\tr_{q/p}(x)=x+x^p+\cdots+x^{p^{n-1}},\ x\in \gf_q.$$
An \emph{additive character} of $\gf_q$ is defined as a homomorphism $\chi$ from $\gf_q$ to $\mathbb{C}^*$ such that  $\chi(x+y)=\chi(x)\chi(y)$ for any $x,y\in \gf_q$. Then
$$\chi_a(x)=\zeta_{p}^{\tr_{q/p}(ax)},\ x\in \gf_q,$$
defines an additive character for each $a\in \gf_q$.
By definition, $\chi_a(x)=\chi_1(ax)$. $\chi_0$  and $\chi_1$ are called the trivial and canonical additive characters, respectively. 
The \emph{complex conjugate} $\overline{\chi_a}$ of $\chi_a$ is defined by $\overline{\chi_a}(x)=\overline{\chi_a(x)}=\chi_a(-x)$ for $x\in \gf_q$.
From \cite{Lidl}, the orthogonal relation of additive characters is given by
\begin{eqnarray*}
\sum_{x\in \gf_q}\chi_1(ax)=\left\{
\begin{array}{ll}
q  &   \mbox{if $a=0$},\\
0    &   \mbox{if $a\in \gf_q^*$}.
\end{array} \right.
\end{eqnarray*}

Let  $\alpha$ be a primitive element of $\gf_q$ and $\zeta_{q-1}$ be the primitive $(q-1)$-th root of complex unity.
A  \emph{multiplicative character} of $\gf_q$ is defined as   a homomorphism $\varphi$ from $\gf_q^*$ to $\mathbb{C}^*$ such that  $\varphi(xy)=\varphi(x)\varphi(y)$ for any $x,y\in \gf_q^*$.
Then 
$$\varphi_j(\alpha^k)=\zeta_{q-1}^{jk},\ k=0,1,\cdots,q-2,$$
for each $j=0,1,\cdots,q-2$, defines a multiplicative character. In particular, $\varphi_0$  is called the trivial  multiplicative character. For odd $q$, if $j=\frac{q-1}{2}$, then $\eta:=\varphi_{\frac{q-1}{2}}$  is called the quadratic multiplicative character of $\gf_q$. 
The \emph{complex conjugate} $\overline{\varphi}$ of a multiplicative character $\varphi$ is defined by $\overline{\varphi}(x)=\overline{\varphi(x)}=\varphi(x^{-1})$ for $x\in \gf_q^*$.
From \cite{Lidl}, the orthogonal relation of multiplicative characters is given by 
\begin{eqnarray*}
\sum_{x\in \gf_q^*}\varphi_j(x)=\left\{
\begin{array}{ll}
q-1  &   \mbox{if $j=0$},\\
0    &   \mbox{if $1\leq j \leq q-2$}.
\end{array} \right.
\end{eqnarray*}

\subsection{Gaussian sums and character sums over finite fields}
The Gaussian sum $G(\varphi,\chi)$ for a multiplicative character $\varphi$ and an additive character $\chi$ of $\gf_q$ is defined by
$$G(\varphi,\chi)=\sum_{x\in \gf_q^*}\varphi(x)\chi(x).$$
The values of Gaussian sums are different to determine. In some special cases, the explicit values of Gaussian sums are known as follows.
\begin{lemma}\label{quadGuasssum1}\cite{Lidl}
Let $q=p^n$ with odd prime $p$. Then
\begin{eqnarray*}G(\eta,\chi_1)&=&(-1)^{n-1}(\sqrt{-1})^{(\frac{p-1}{2})^2n}\sqrt{q}\\
 &=&\left\{
\begin{array}{lll}
(-1)^{n-1}\sqrt{q}    &   \mbox{for }p\equiv 1\pmod{4},\\
(-1)^{n-1}(\sqrt{-1})^{n}\sqrt{q}    &   \mbox{for }p\equiv 3\pmod{4}.
\end{array} \right. \end{eqnarray*}
\end{lemma}
\begin{lemma}\label{quadGuasssum2}\cite{Lidl}
Let $\varphi$ be a multiplicative character of $\gf_q$ and $\chi$ be an additive character of $\gf_q$. Then  
\begin{eqnarray*}
G(\varphi,\chi)&=&\left\{
\begin{array}{lll}
q-1 & \text{if}\ \varphi=\varphi_{0}, \chi=\chi_{0},\\
-1  &  \text{if}\ \varphi=\varphi_{0}, \chi \neq \chi_{0},\\
0  &  \text{if}\ \varphi \neq \varphi_{0}, \chi=\chi_{0}.
\end{array}\right. \end{eqnarray*}
If $\varphi\neq\varphi_{0}$ and $\chi\neq\chi_{0}$, then $|G(\varphi,\chi_1)|=\sqrt{q}$. 

\end{lemma}

The following character sums will be used in this paper. 
 \begin{lemma}\label{lem-charactersum}\cite{Lidl}
Let $\chi$ be a nontrivial additive character of $\gf_q$ and $q$ be a power of an odd prime. Let $f(x)=a_2x^2+a_1x+a_0\in \gf_q[x]$ with $a_2\neq 0$. Then
$$\sum_{x\in \gf_q}\chi(f(x))=\chi(a_0-a_1^2(4a_2)^{-1})\eta(a_2)G(\eta,\chi).$$
\end{lemma}

\subsection{The number of zeros of some polynomials over finite fields}
In the following, we present some known results on the numbers of zeros of some polynomials over finite fields.
\begin{lemma}\label{lem-x^2root}\cite{KP}
Let $\gf_q$ be a finite field of characteristic $2$ and let $f(x)=ax^2+bx+c\in \gf_q[x]$ be a polynomial of degree $2$. Then 
\begin{enumerate}
\item $f$ has exactly one zero in $\gf_q$ if and only if $b=0$;
\item $f$ has exactly two zeros in $\gf_q$ if and only if $b\neq0$ and $\tr_{q/2}(\frac{ac}{b^2})=0$;
\item $f$ has no zero in $\gf_q$ if and only if $b\neq0$ and $\tr_{q/2}(\frac{ac}{b^2})=1$.
\end{enumerate}
\end{lemma}

Now we consider the polynomial
\begin{eqnarray*}
f(x)=x^{p^k+1}+ax^{p^k}+bx+c\in \Bbb F_q[x].
\end{eqnarray*}
For $a\neq0$, we substitute $x$ by $x-a$ in $f(x)$. Let $\alpha_1:=b-a^{p^k}$ and $\beta_1:=c-ab$.
Then $f(x)$ can be reduced to $g(x)=x^{p^k+1}+\alpha_1 x +\beta_1$.
If we substitute $x$ by $ux$ in $g(x)$ for  $u^{p^k}=\alpha$ and $\alpha=b-a^{p^k}\neq0$, then $g(x)$  can be transformed into the form $P_\gamma(x)=x^{p^k+1}+x+\gamma$,
 where $\gamma=\frac{\beta}{u^{p^k+1}}$. By the results in \cite{B}, we can easily derive the number of zeros of $f(x)$ in $\gf_q$. 
 
\begin{lemma}\label{1111}\cite{XG}
Let $k$ and $n$ be positive integers with $\gcd(k,n)=h$. Let $f(x)=x^{p^k+1}+ax^{p^k}+bx+c$, where $a, b ,c\in \gf_q$. The number of zeros of $f(x)$ in $\gf_q$ is denoted by $N_f$.
\begin{enumerate}
\item If $a=0, b=0$ or $a\neq0, b=a^{p^k}$, then
\begin{eqnarray*}
N_f=
\left\{
\begin{array}{ll}
1 & \text{if} ~\frac{n}{h}~ \text{is}~ \text{odd} ~\text{and}~ p=2,\\
0, 1 ~\text{or} ~2 & \text{if} ~\frac{n}{h} ~\text{is} ~\text{odd} ~\text{and} ~p ~\text{is} ~\text{odd},\\
0, 1 ~\text{or}~ p^h+1 &  \text{if}~ \frac{n}{h} ~\text{is}~ \text{even}.
\end{array}\right.
\end{eqnarray*}
\item If $a=0, b\neq0$ or $a\neq0, b\neq a^{p^k}$, then $N_f$ takes either $0, 1, 2$ or $p^{h}+1$.
\end{enumerate}
\end{lemma}

\begin{corollary}\label{cor}
Let $h=1$, $n\geq 3$ and $p=2$ in Lemma \ref{1111}. Then $N_f\in \{0,1,2,3\}$ and there exists $(a,b,c)\in \gf_q^3$ such that $N_f=3$.
\end{corollary}

\begin{proof}
By Lemma \ref{1111}, we have $N_f\in \{0,1,2,3\}$.
According to \cite[Theorems 8-10]{KHM} and the relationship between $f(x)$ and $P_\gamma(x)$, it is easy to deduce that there exists $(a,b,c)\in \gf_q^3$ such that $N_f=3$. 
\end{proof}

\begin{lemma} \label{1112}\cite{XG}
Let $q=p^n$ and $k$ be a positive integer with $\gcd(k,n)=h$. Then the trinomial $x^{p^k}-ax-b$ has either none, one, or $p^h$ zeros in $\gf_q$, where $a,b\in \gf_q$.   
\end{lemma}

\section{New constructions of periodic QCSSs}
In this section, we will present several new constructions of asymptotically optimal or near-optimal periodic QCSSs with large set sizes from special polynomials over finite fields or Gaussian sums.

\subsection{The first construction of periodic QCSSs from quadratic polynomials}\label{secA}
Let $q=p^{n}$ for a prime $p$ and a positive integer $n$. Let $\alpha$ be a primitive element of $\gf_q$. 
Assume that $d_0, d_1,\cdots, d_{q-1}$ are all the elements of $\gf_q$, i.e.  $\gf_q=\{d_0, d_1,\cdots, d_{q-1}\}$.
  Let $f(x)=x^2+ax+b\in \gf_q[x]$.
Let $\chi_{1}$ be the canonical additive character of $\gf_q$. We define a constituent sequence $\mathbf{c}_{l}^{a,b}$ of period $q-1$ as
\begin{eqnarray*}
\mathbf{c}_{l}^{a,b}=\left(\mathbf{c}_{l}^{a,b}(t)\right)_{t=0}^{q-2}, \mbox{ where}\ \mathbf{c}_{l}^{a,b}(t)=\chi_{1}(\alpha^{t}f(d_l))\mbox{ for} ~ 0\leq l\leq q-1.
\end{eqnarray*}
Then we obtain a two-dimensional $q\times(q-1)$ matrix 
\begin{eqnarray*}
\mathbf{C}^{a,b}=\left[
\begin{array}{cccc}
\mathbf{c}_{0}^{a,b}\\
\mathbf{c}_{1}^{a,b}\\
\vdots,\\
\mathbf{c}_{q-1}^{a,b}
\end{array}\right].
\end{eqnarray*}
Such matrixes yield a complementary sequence set given by
\begin{eqnarray}\label{eq-c1}
\mathcal{C}=\left\{\mathbf{C}^{a,b}: a\in\gf_q, b\in \gf_q\right\}.
\end{eqnarray}

Firstly, we study the case that $p$ is an odd prime.
\begin{theorem}\label{111}
Let $q=p^n$ for odd prime $p$ and  positive integer $n$. Let $\mathcal{C}$ be the complementary sequence set
defined in  (\ref{eq-c1}). Then $\mathcal{C}$ is a periodic $(q^2, q, q-1, q)$-QCSS with alphabet size $p$ which is asymptotically optimal with respect
to the lower bound in (\ref{wlech}).
\end{theorem}
\begin{proof}
For any two  complementary sequences $\mathbf{C}^{a_{1},b_{1}}$ and $\mathbf{C}^{a_{2},b_{2}}$ of $\mathcal{C}$ and $\tau\in[0,q-2]$, we have
\begin{eqnarray*}
\nonumber &
&R_{\mathbf{C}^{a_{1},b_{1}},\mathbf{C}^{a_{2},b_{2}}}(\tau)\\
\nonumber 
&=&\sum_{l=0}^{q-1}R_{\mathbf{c}_{l}^{a_{1},b_{1}},\mathbf{c}_{l}^{a_{2},b_{2}}}(\tau)
\\
\nonumber 
&=&\sum_{l=0}^{q-1}\sum_{t=0}^{q-2}\chi_{1}(\alpha^{t}(d_{l}^2+a_{1}d_{l}+b_{1}))\overline{\chi_{1}}(\alpha^{t+\tau}(d_{l}^2+a_{2}d_{l}+b_{2}))
\\
&=&\sum_{l=0}^{q-1}\sum_{t=0}^{q-2}\chi_{1}((1-\alpha^{\tau})\alpha^{t}d_{l}^{2}+(a_{1}-a_{2}\alpha^{\tau})\alpha^{t}d_{l}+(b_{1}-b_{2}\alpha^{\tau})\alpha^{t}),
\end{eqnarray*}
where $a_1,b_1,a_2,b_2\in \gf_q$.
We now consider the following cases to determine the value distribution of $R_{\mathbf{C}^{a_{1},b_{1}},\mathbf{C}^{a_{2},b_{2}}}(\tau)$.

{Case 1}: If $\tau=0$, $a_1=a_2$ and $b_1\neq b_2$, then 
\begin{eqnarray}
\nonumber &
&R_{\mathbf{C}^{a_{1},b_{1}},\mathbf{C}^{a_{2},b_{2}}}(\tau)\\
\nonumber 
&=&\sum_{l=0}^{q-1}\sum_{t=0}^{q-2}\chi_{1}((b_{1}-b_{2})\alpha^{t}).
\end{eqnarray}
By the orthogonal relation of the additive characters, we have
\begin{eqnarray*}
R_{\mathbf{C}^{a_{1},b_{1}},\mathbf{C}^{a_{2},b_{2}}}(\tau)=-q.
\end{eqnarray*}

{Case 2}: If $\tau=0$, $a_1\neq a_2$ and $b_1= b_2$, then 
\begin{eqnarray*}
\nonumber &
&R_{\mathbf{C}^{a_{1},b_{1}},\mathbf{C}^{a_{2},b_{2}}}(\tau)\\
\nonumber 
&=&\sum_{l=0}^{q-1}\sum_{t=0}^{q-2}\chi_{1}((a_{1}-a_{2})\alpha^{t}d_{l})\\
&=&\sum_{t=0}^{q-2}\sum_{x\in \gf_q}\chi_{1}((a_{1}-a_{2})\alpha^{t}x).
\end{eqnarray*}
By the orthogonal relation of the additive characters, we have
\begin{eqnarray*}
R_{\mathbf{C}^{a_{1},b_{1}},\mathbf{C}^{a_{2},b_{2}}}(\tau)=0.
\end{eqnarray*}

{Case 3}: If $\tau=0$, $a_1\neq a_2$ and $b_1\neq b_2$, then 
\begin{eqnarray*}
\nonumber &
&R_{\mathbf{C}^{a_{1},b_{1}},\mathbf{C}^{a_{2},b_{2}}}(\tau)\\
\nonumber 
&=&\sum_{l=0}^{q-1}\sum_{t=0}^{q-2}\chi_{1}((a_{1}-a_{2})\alpha^{t}d_{l})\chi_{1}((b_{1}-b_{2})\alpha^{t})\\
&=&\sum_{t=0}^{q-2}\chi_{1}((b_{1}-b_{2})\alpha^{t})\sum_{x\in \gf_q}\chi_{1}((a_{1}-a_{2})\alpha^{t}x).
\end{eqnarray*}
By the orthogonal relation of the additive characters, we have
\begin{eqnarray*}
R_{\mathbf{C}^{a_{1},b_{1}},\mathbf{C}^{a_{2},b_{2}}}(\tau)=0.
\end{eqnarray*}

{Case 4}: If  $\tau\neq0$, then $1-\alpha^{\tau}\neq0$. According to Lemma \ref{lem-charactersum}, we have 
\begin{eqnarray*}
\nonumber &
&R_{\mathbf{C}^{a_{1},b_{1}},\mathbf{C}^{a_{2},b_{2}}}(\tau)\\
\nonumber
&=&\sum_{t=0}^{q-2}\sum_{x\in \gf_q}\chi_{1}((1-\alpha^{\tau})\alpha^{t}x^{2}+(a_{1}-a_{2}\alpha^{\tau})\alpha^{t}x+(b_{1}-b_{2}\alpha^{\tau})\alpha^{t})\\
&=&\sum_{t=0}^{q-2}\chi_{1}\left((b_{1}-b_{2}\alpha^{\tau})\alpha^{t}-\frac{((a_{1}-a_{2}\alpha^{\tau})\alpha^{t})^2}{4(1-\alpha^{\tau})\alpha^{t}}\right)\eta((1-\alpha^{\tau})\alpha^{t})G(\eta, \chi_{1})\\
&=&G(\eta, \chi_{1})\sum_{x\in \gf_q^*}\chi_{1}\left(\left(b_{1}-b_{2}\alpha^{\tau}-\frac{(a_{1}-a_{2}\alpha^{\tau})^2}{4(1-\alpha^{\tau})}\right)x\right)\eta((1-\alpha^{\tau})x).
\end{eqnarray*}
If $b_{1}-b_{2}\alpha^{\tau}=\frac{(a_{1}-a_{2}\alpha^{\tau})^2}{4(1-\alpha^{\tau})}$, by the orthogonal relation of the  multiplicative characters, we have
$$R_{\mathbf{C}^{a_{1},b_{1}},\mathbf{C}^{a_{2},b_{2}}}(\tau)=0.$$
If $b_{1}-b_{2}\alpha^{\tau}\neq\frac{(a_{1}-a_{2}\alpha^{\tau})^2}{4(1-\alpha^{\tau})}$, then
\begin{eqnarray*}
& &R_{\mathbf{C}^{a_{1},b_{1}},\mathbf{C}^{a_{2},b_{2}}}(\tau)\\
&=&G(\eta, \chi_{1})\eta(1-\alpha^{\tau})\eta\left(b_{1}-b_{2}\alpha^{\tau}-\frac{(a_{1}-a_{2}\alpha^{\tau})^2}{4(1-\alpha^{\tau})}\right) \\
& &\cdot \sum_{x\in \gf_q^*}\chi_{1}\left(\left(b_{1}-b_{2}\alpha^{\tau}-\frac{(a_{1}-a_{2}\alpha^{\tau})^2}{4(1-\alpha^{\tau})}\right)x\right)
\eta\left(\left(b_{1}-b_{2}\alpha^{\tau}-\frac{(a_{1}-a_{2}\alpha^{\tau})^2}{4(1-\alpha^{\tau})}\right)x\right)\\
&=&G(\eta, \chi_{1})^2\eta(1-\alpha^{\tau})\eta\left(b_{1}-b_{2}\alpha^{\tau}-\frac{(a_{1}-a_{2}\alpha^{\tau})^2}{4(1-\alpha^{\tau})}\right),
\end{eqnarray*}
where $G(\eta, \chi_1) = (-1)^{n-1}(\sqrt{-1})^{(\frac{p-1}{2})^2n}\sqrt{q}$  by Lemma \ref{quadGuasssum1}. Thus,
$ |R_{\mathbf{C}^{a_{1},b_{1}},\mathbf{C}^{a_{2},b_{2}}}(\tau)|\in \{0,q\}$.

Summarizing the above four cases, we derive that the maximum periodic correlation magnitude of $\mathcal{C}$ is $q$. Next, we will show that  the parameters of the obtained periodic QCSS asymptotically achieve the correlation lower bound in (\ref{wlech}). Since $\mathcal{C}$ is a periodic $(q^2, q, q-1, q)$-QCSS, according to the bound in (\ref{wlech}), we have 
\begin{eqnarray*}
\vartheta_{\opt}=q(q-1)\sqrt{\frac{q-1}{q^2(q-1)-1)}}=q\sqrt{\frac{(q-1)^{3}}{q^2(q-1)-1)}}.
\end{eqnarray*}
It is easy to see that 
\begin{eqnarray*}
\lim_{q\rightarrow+\infty}\frac{\vartheta_{\max}}{\vartheta_{\opt}}=\lim_{q\rightarrow+\infty}\frac{q}{q\sqrt{\frac{(q-1)^{3}}{q^2(q-1)-1)}}}=1.
\end{eqnarray*}
This completes the proof of this theorem.

\end{proof}

\begin{example}\label{example1}
  Let $p=5$ and $n=2$. Then the parameters of the QCSS $\mathcal{C}$ constructed in Theorem \ref{111} are $(625, 25, 24, 25)$ and its alphabet is given hy $\{e^{2\pi\sqrt{-1}i/5}: i\in[0,4]\}$. 
  By magma program, we obtain that the periodic correlation function $R_{\mathbf{C}^{a_{1},b_{1}},\mathbf{C}^{a_{2},b_{2}}}(\tau)$ is equal to $0$ or $25$ for any $a_1,b_1,a_2,b_2\in \gf_q$ and $0\leq\tau\leq23$ except the trivial case that $a_1=a_2, b_1=b_2$ and $\tau=0$. For instance, the matrices $\mathbf{C}^{1,\alpha^0}$ and $\mathbf{C}^{1,\alpha^6}$ in $\mathcal{C}$ are presented as follows, where each entry stands for a power of $\zeta_5=e^{2\pi\sqrt{-1}/5} $. By Python program, we show the autocorrelation magnitude distribution of $\mathbf{C}^{1,\alpha^0}$ in Fig. \ref{fig1},  the autocorrelation magnitude distribution of $\mathbf{C}^{1,\alpha^6}$ in Fig. \ref{fig2}, and the correlation magnitude distribution of $\mathbf{C}^{1,\alpha^0}$ and $\mathbf{C}^{1,\alpha^6}$ in Fig. \ref{fig3}, respectively.  

\begin{eqnarray*}\label{matrix1}
\mathbf{C}^{1,\alpha^0}=\left[
\begin{array}{cccc}
3 1 0 3 3 2 1 2 0 1 1 4 2 4 0 2 2 3 4 3 0 4 4 1\\
3 3 2 1 2 0 1 1 4 2 4 0 2 2 3 4 3 0 4 4 1 3 1 0\\
2 2 3 4 3 0 4 4 1 3 1 0 3 3 2 1 2 0 1 1 4 2 4 0\\
4 2 4 0 2 2 3 4 3 0 4 4 1 3 1 0 3 3 2 1 2 0 1 1\\
2 3 4 3 0 4 4 1 3 1 0 3 3 2 1 2 0 1 1 4 2 4 0 2\\
2 2 3 4 3 0 4 4 1 3 1 0 3 3 2 1 2 0 1 1 4 2 4 0\\
2 4 0 2 2 3 4 3 0 4 4 1 3 1 0 3 3 2 1 2 0 1 1 4\\
0 1 1 4 2 4 0 2 2 3 4 3 0 4 4 1 3 1 0 3 3 2 1 2\\
0 0 0 0 0 0 0 0 0 0 0 0 0 0 0 0 0 0 0 0 0 0 0 0\\
1 0 3 3 2 1 2 0 1 1 4 2 4 0 2 2 3 4 3 0 4 4 1 3\\
3 3 2 1 2 0 1 1 4 2 4 0 2 2 3 4 3 0 4 4 1 3 1 0\\
2 3 4 3 0 4 4 1 3 1 0 3 3 2 1 2 0 1 1 4 2 4 0 2\\
1 2 0 1 1 4 2 4 0 2 2 3 4 3 0 4 4 1 3 1 0 3 3 2\\
4 3 0 4 4 1 3 1 0 3 3 2 1 2 0 1 1 4 2 4 0 2 2 3\\
2 0 1 1 4 2 4 0 2 2 3 4 3 0 4 4 1 3 1 0 3 3 2 1\\
2 0 1 1 4 2 4 0 2 2 3 4 3 0 4 4 1 3 1 0 3 3 2 1\\
0 0 0 0 0 0 0 0 0 0 0 0 0 0 0 0 0 0 0 0 0 0 0 0\\
4 3 0 4 4 1 3 1 0 3 3 2 1 2 0 1 1 4 2 4 0 2 2 3\\
3 1 0 3 3 2 1 2 0 1 1 4 2 4 0 2 2 3 4 3 0 4 4 1\\
2 1 2 0 1 1 4 2 4 0 2 2 3 4 3 0 4 4 1 3 1 0 3 3\\
0 1 1 4 2 4 0 2 2 3 4 3 0 4 4 1 3 1 0 3 3 2 1 2\\
2 1 2 0 1 1 4 2 4 0 2 2 3 4 3 0 4 4 1 3 1 0 3 3\\
4 2 4 0 2 2 3 4 3 0 4 4 1 3 1 0 3 3 2 1 2 0 1 1\\
1 0 3 3 2 1 2 0 1 1 4 2 4 0 2 2 3 4 3 0 4 4 1 3\\
1 2 0 1 1 4 2 4 0 2 2 3 4 3 0 4 4 1 3 1 0 3 3 2\\
\end{array}\right],\
\mathbf{C}^{1,\alpha^6}=\left[
\begin{array}{cccc}
4 3 0 4 4 1 3 1 0 3 3 2 1 2 0 1 1 4 2 4 0 2 2 3\\
4 0 2 2 3 4 3 0 4 4 1 3 1 0 3 3 2 1 2 0 1 1 4 2\\
3 4 3 0 4 4 1 3 1 0 3 3 2 1 2 0 1 1 4 2 4 0 2 2\\
0 4 4 1 3 1 0 3 3 2 1 2 0 1 1 4 2 4 0 2 2 3 4 3\\
3 0 4 4 1 3 1 0 3 3 2 1 2 0 1 1 4 2 4 0 2 2 3 4\\
3 4 3 0 4 4 1 3 1 0 3 3 2 1 2 0 1 1 4 2 4 0 2 2\\
3 1 0 3 3 2 1 2 0 1 1 4 2 4 0 2 2 3 4 3 0 4 4 1\\
1 3 1 0 3 3 2 1 2 0 1 1 4 2 4 0 2 2 3 4 3 0 4 4\\
1 2 0 1 1 4 2 4 0 2 2 3 4 3 0 4 4 1 3 1 0 3 3 2\\
2 2 3 4 3 0 4 4 1 3 1 0 3 3 2 1 2 0 1 1 4 2 4 0\\
4 0 2 2 3 4 3 0 4 4 1 3 1 0 3 3 2 1 2 0 1 1 4 2\\
3 0 4 4 1 3 1 0 3 3 2 1 2 0 1 1 4 2 4 0 2 2 3 4\\
2 4 0 2 2 3 4 3 0 4 4 1 3 1 0 3 3 2 1 2 0 1 1 4\\
0 0 0 0 0 0 0 0 0 0 0 0 0 0 0 0 0 0 0 0 0 0 0 0\\
3 2 1 2 0 1 1 4 2 4 0 2 2 3 4 3 0 4 4 1 3 1 0 3\\
3 2 1 2 0 1 1 4 2 4 0 2 2 3 4 3 0 4 4 1 3 1 0 3\\
1 2 0 1 1 4 2 4 0 2 2 3 4 3 0 4 4 1 3 1 0 3 3 2\\
0 0 0 0 0 0 0 0 0 0 0 0 0 0 0 0 0 0 0 0 0 0 0 0\\
4 3 0 4 4 1 3 1 0 3 3 2 1 2 0 1 1 4 2 4 0 2 2 3\\
3 3 2 1 2 0 1 1 4 2 4 0 2 2 3 4 3 0 4 4 1 3 1 0\\
1 3 1 0 3 3 2 1 2 0 1 1 4 2 4 0 2 2 3 4 3 0 4 4\\
3 3 2 1 2 0 1 1 4 2 4 0 2 2 3 4 3 0 4 4 1 3 1 0\\
0 4 4 1 3 1 0 3 3 2 1 2 0 1 1 4 2 4 0 2 2 3 4 3\\
2 2 3 4 3 0 4 4 1 3 1 0 3 3 2 1 2 0 1 1 4 2 4 0\\
2 4 0 2 2 3 4 3 0 4 4 1 3 1 0 3 3 2 1 2 0 1 1 4\\
\end{array}\right].
\end{eqnarray*}

\begin{figure}[htbp]
\centering
\includegraphics[width=0.6\columnwidth,height=0.4\linewidth]{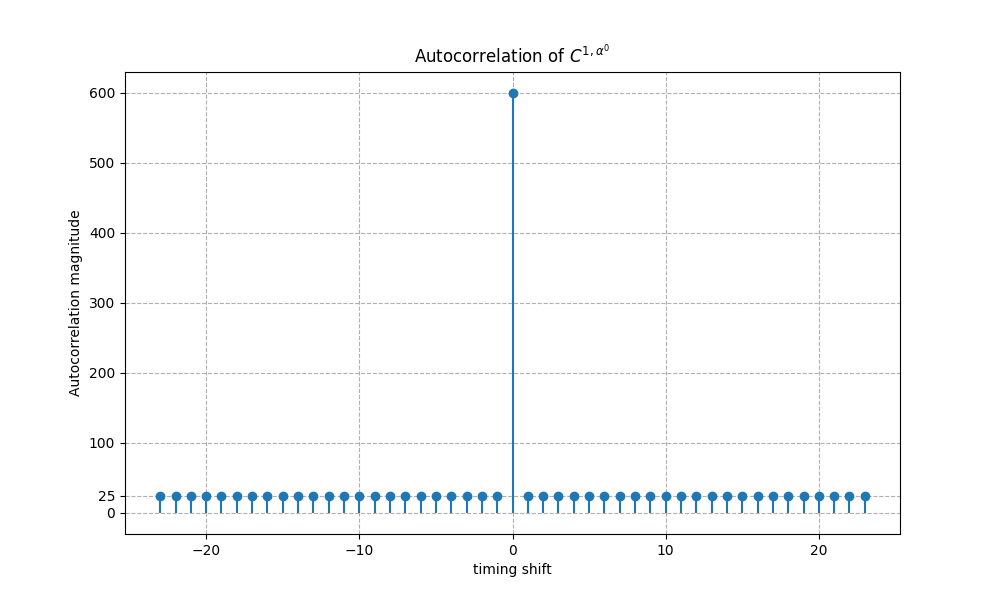}
\caption{The autocorrelation magnitude  distribution of $\mathbf{C}^{1,\alpha^0}$ in Example \ref{example1}}
\label{fig1}
\end{figure}

\begin{figure}[htbp]
\centering
\includegraphics[width=0.6\columnwidth,height=0.4\linewidth]{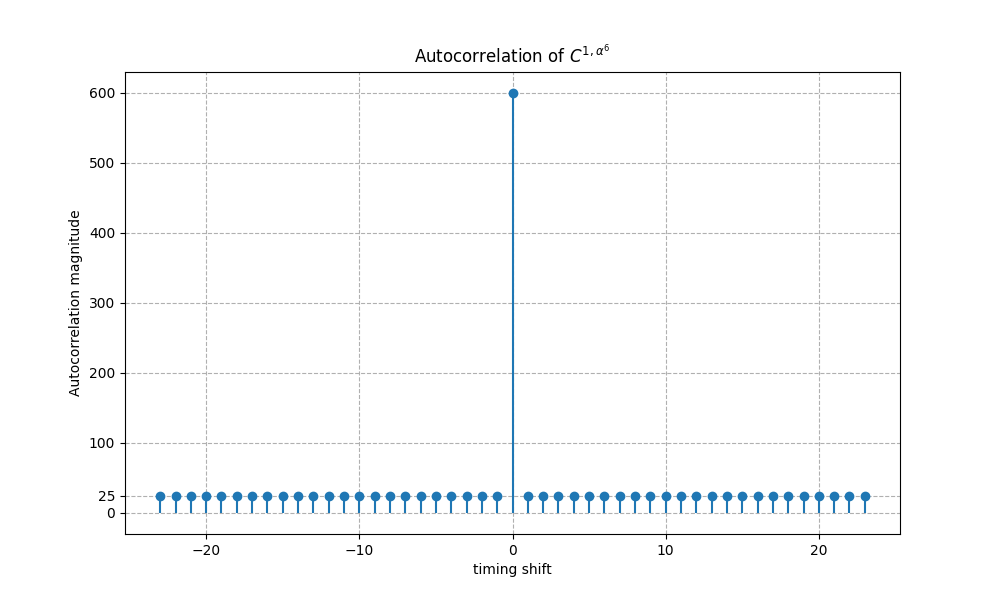}
\caption{\centering{The autocorrelation magnitude distribution of $\mathbf{C}^{1,\alpha^6}$} in Example \ref{example1}}
\label{fig2}
\end{figure}

\begin{figure}[htbp]
\centering
\includegraphics[width=0.6\columnwidth,height=0.4\linewidth]{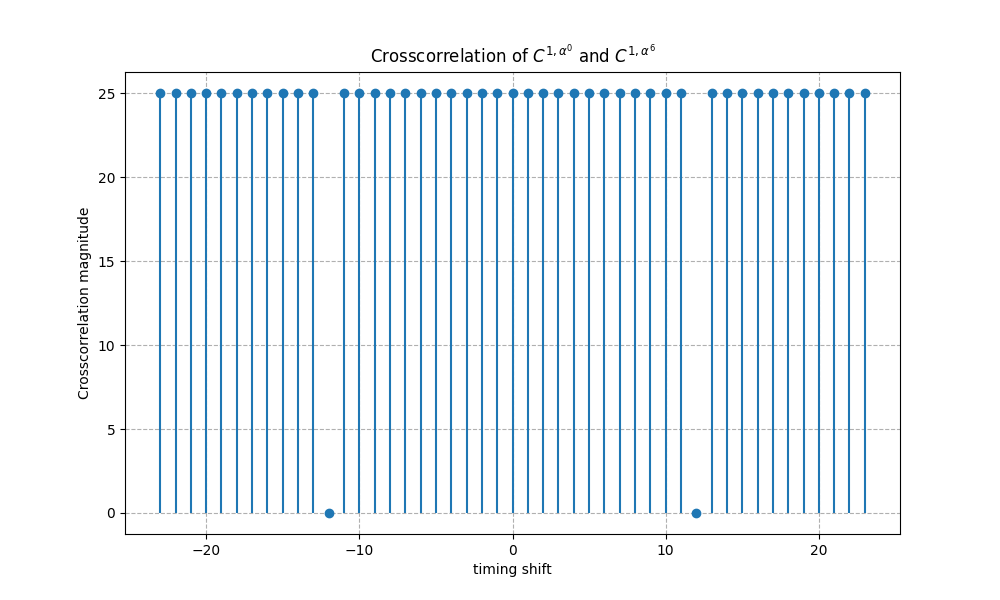}
\caption{\centering The correlation magnitude distribution of $\mathbf{C}^{1,\alpha^0}$ and $\mathbf{C}^{1,\alpha^6}$ in Example \ref{example1}}
\label{fig3}
\end{figure}

\end{example}
Secondly, we consider the case that $p=2$ in the following.
\begin{theorem}\label{222}
Let $q=2^n$ for a positive integer $n$. Let $\mathcal{C}$ be the complementary sequence set defined in (\ref{eq-c1}). Then $\mathcal{C}$ is a periodic $(q^2, q, q-1, q)$-QCSS with alphabet size $2$ which is asymptotically optimal with respect
to the correlation lower bound in (\ref{wlech}).
\end{theorem}
\begin{proof}
For any two  complementary sequences $\mathbf{C}^{a_{1},b_{1}}$, $\mathbf{C}^{a_{2},b_{2}}$ in $\mathcal{C}$ and $\tau\in[0,q-2]$, we have
\begin{eqnarray*}
\nonumber &
&R_{\mathbf{C}^{a_{1},b_{1}},\mathbf{C}^{a_{2},b_{2}}}(\tau)\\
\nonumber 
&=&\sum_{l=0}^{q-1}R_{\mathbf{c}_{l}^{a_{1},b_{1}},\mathbf{c}_{l}^{a_{2},b_{2}}}(\tau)
\\
\nonumber 
&=&\sum_{l=0}^{q-1}\sum_{t=0}^{q-2}\chi_{1}\left(\alpha^{t}(d_{l}^2+a_{1}d_{l}+b_{1})\right)\overline{\chi_{1}}\left(\alpha^{t+\tau}(d_{l}^2+a_{2}d_{l}+b_{2})\right)
\\
&=&\sum_{l=0}^{q-1}\sum_{t=0}^{q-2}\chi_{1}\left((1-\alpha^{\tau})\alpha^{t}d_{l}^{2}+(a_{1}-a_{2}\alpha^{\tau})\alpha^{t}d_{l}+(b_{1}-b_{2}\alpha^{\tau})\alpha^{t}\right),
\end{eqnarray*}
where $a_1,b_1,a_2,b_2\in \gf_q$.
Then we consider the following cases to determine the value distribution  of $R_{\mathbf{C}^{a_{1},b_{1}},\mathbf{C}^{a_{2},b_{2}}}(\tau)$.

Case 1: If $\tau=0$, similarly to the proof of Theorem \ref{111}, we have
 \begin{eqnarray*}
R_{\mathbf{C}^{a_{1},b_{1}},\mathbf{C}^{a_{2},b_{2}}}(\tau)=
\left\{
\begin{array}{ll}
-q & \text{if} ~ a_1=a_2, b_1\neq b_2,\\
0 & \text{if} ~ a_1\neq a_2,b_1=b_2,\\
0 &  \text{if} ~ a_1\neq a_2, b_1\neq b_2.
\end{array}\right.
\end{eqnarray*}

Case 2: If $\tau\neq0$ and $a_1=a_2\alpha^{\tau}$, then
\begin{eqnarray*}
\nonumber &
&R_{\mathbf{C}^{a_{1},b_{1}},\mathbf{C}^{a_{2},b_{2}}}(\tau)\\
\nonumber 
&=&\sum_{l=0}^{q-1}\sum_{t=0}^{q-2}\chi_{1}\left((1-\alpha^{\tau})\alpha^{t}d_{l}^{2}+(b_{1}-b_{2}\alpha^{\tau})\alpha^{t}\right)\\
&=&\sum_{t=0}^{q-2}\sum_{x\in \gf_q}\chi_{1}\left(\left((1-\alpha^{\tau})x^{2}+(b_{1}-b_{2}\alpha^{\tau})\right)\alpha^{t}\right)\\
&=&\sum_{x\in \gf_q}\sum_{y\in \gf_q^*}\chi_{1}\left(\left((1-\alpha^{\tau})x^{2}+(b_{1}-b_{2}\alpha^{\tau})\right)y\right).
\end{eqnarray*}
By Lemma \ref{lem-x^2root}, we know that $f(x)=(1-\alpha^{\tau})x^{2}+b_{1}-b_{2}\alpha^{\tau}$ has exactly one zero in $\gf_q$. By the orthogonal relation of the additive characters, we have
\begin{eqnarray*}
\nonumber &
&R_{\mathbf{C}^{a_{1},b_{1}},\mathbf{C}^{a_{2},b_{2}}}(\tau)\\
\nonumber 
&=&(q-1)+(q-1)(-1)\\
&=&0.
\end{eqnarray*}

Case 3: If $\tau\neq0$ and $a_1\neq a_2\alpha^{\tau}$, then
\begin{eqnarray*}
\nonumber &
&R_{\mathbf{C}^{a_{1},b_{1}},\mathbf{C}^{a_{2},b_{2}}}(\tau)\\
\nonumber 
&=&\sum_{x\in \gf_q}\sum_{y\in \gf_q^*}\chi_{1}\left((1-\alpha^{\tau})yx^{2}+(a_{1}-a_{2}\alpha^{\tau})yx+(b_{1}-b_{2}\alpha^{\tau})y\right)\\
&=&\sum_{x\in \gf_q}\sum_{y\in \gf_q^*}\chi_{1}\left(\left((1-\alpha^{\tau})x^{2}+(a_{1}-a_{2}\alpha^{\tau})x+(b_{1}-b_{2}\alpha^{\tau})\right)y\right).
\end{eqnarray*}
By Lemma \ref{lem-x^2root}, we know that $f(x)=(1-\alpha^{\tau})x^{2}+(a_{1}-a_{2}\alpha^{\tau})x+(b_{1}-b_{2}\alpha^{\tau})$ has exactly two zeros if $\tr_{q/2}\left(\frac{(1-\alpha^{\tau})(b_1-b_2\alpha^{\tau})}{(a_1-a_2\alpha^{\tau})^2}\right)=0$, and no zero if $\tr_{q/2}\left(\frac{(1-\alpha^{\tau})(b_1-b_2\alpha^{\tau})}{(a_1-a_2\alpha^{\tau})^2}\right)=1$.
By the orthogonal relation of the additive characters, we have
\begin{eqnarray*}
R_{\mathbf{C}^{a_{1},b_{1}},\mathbf{C}^{a_{2},b_{2}}}(\tau)=
\left\{
\begin{array}{ll}
-q & \text{if} ~ \tr_{q/2}\left(\frac{(1-\alpha^{\tau})(b_1-b_2\alpha^{\tau})}{(a_1-a_2\alpha^{\tau})^2}\right)=1,\\
q & \text{if} ~ \tr_{q/2}\left(\frac{(1-\alpha^{\tau})(b_1-b_2\alpha^{\tau})}{(a_1-a_2\alpha^{\tau})^2}\right)=0.
\end{array}\right.
\end{eqnarray*}

Based on all the cases discussed above, we deduce that the maximum periodic correlation magnitude of $\mathcal{C}$ is $q$.
Then $\mathcal{C}$ is a periodic $(q^2, q, q-1, q)$-QCSS which is asymptotically optimal with respect to the correlation bound in  (\ref{wlech}) by a similar proof as that of 
Theorem \ref{111}. This completes the proof of this theorem.
\end{proof}
\begin{example}\label{example2}
Let $p=2$ and $n=4$. Then the parameters of the QCSS $\mathcal{C}$ constructed in Theorem \ref{222} are $(256, 16, 15, 16)$ and its alphabet is given by $\{(-1)^i: i\in[0,1]\}$. By Magma program, we verify that the periodic correlation function $R_{\mathbf{C}^{a_{1},b_{1}},\mathbf{C}^{a_{2},b_{2}}}(\tau)$ is equal to $0$ or $16$ for any $a_1,b_1,a_2,b_2\in \gf_q$ and $0\leq\tau\leq14$ except the trivial case that $a_1=a_2, b_1=b_2$ and $\tau=0$. For instance, the matrices $\mathbf{C}^{1,\alpha}$ and $\mathbf{C}^{1,\alpha^2}$ in $\mathcal{C}$ are presented as follows, where each entry stands for a power of $\zeta_2=-1$. By Python program, we show the autocorrelation magnitude distribution of $\mathbf{C}^{1,\alpha}$ in Fig. \ref{fig4},  the autocorrelation magnitude distribution of $\mathbf{C}^{1,\alpha^2}$ in Fig. \ref{fig5}, and the correlation magnitude distribution of $\mathbf{C}^{1,\alpha}$ and $\mathbf{C}^{1,\alpha^2}$ in Fig. \ref{fig6}, respectively.  
\begin{eqnarray*}\label{matrix2}
\mathbf{C}^{1,\alpha}=\left[
\begin{array}{cccc}
0 1 0 0 1 1 0 1 0 1 1 1 1 0 0\\
1 0 0 1 1 0 1 0 1 1 1 1 0 0 0\\
1 0 1 1 1 1 0 0 0 1 0 0 1 1 0\\
1 1 0 1 0 1 1 1 1 0 0 0 1 0 0\\
1 0 0 1 1 0 1 0 1 1 1 1 0 0 0\\
0 1 1 0 1 0 1 1 1 1 0 0 0 1 0\\
0 0 1 0 0 1 1 0 1 0 1 1 1 1 0\\
0 0 0 0 0 0 0 0 0 0 0 0 0 0 0\\
1 0 1 1 1 1 0 0 0 1 0 0 1 1 0\\
0 0 0 0 0 0 0 0 0 0 0 0 0 0 0\\
0 1 1 0 1 0 1 1 1 1 0 0 0 1 0\\
1 1 1 1 0 0 0 1 0 0 1 1 0 1 0\\
1 1 1 1 0 0 0 1 0 0 1 1 0 1 0\\
0 0 1 0 0 1 1 0 1 0 1 1 1 1 0\\
1 1 0 1 0 1 1 1 1 0 0 0 1 0 0\\
0 1 0 0 1 1 0 1 0 1 1 1 1 0 0\\
\end{array}\right],\
\mathbf{C}^{1,\alpha^2}=\left[
\begin{array}{cccc}
1 0 0 1 1 0 1 0 1 1 1 1 0 0 0\\
0 1 0 0 1 1 0 1 0 1 1 1 1 0 0\\
0 1 1 0 1 0 1 1 1 1 0 0 0 1 0\\
0 0 0 0 0 0 0 0 0 0 0 0 0 0 0\\
0 1 0 0 1 1 0 1 0 1 1 1 1 0 0\\
1 0 1 1 1 1 0 0 0 1 0 0 1 1 0\\
1 1 1 1 0 0 0 1 0 0 1 1 0 1 0\\
1 1 0 1 0 1 1 1 1 0 0 0 1 0 0\\
0 1 1 0 1 0 1 1 1 1 0 0 0 1 0\\
1 1 0 1 0 1 1 1 1 0 0 0 1 0 0\\
1 0 1 1 1 1 0 0 0 1 0 0 1 1 0\\
0 0 1 0 0 1 1 0 1 0 1 1 1 1 0\\
0 0 1 0 0 1 1 0 1 0 1 1 1 1 0\\
1 1 1 1 0 0 0 1 0 0 1 1 0 1 0\\
0 0 0 0 0 0 0 0 0 0 0 0 0 0 0\\
1 0 0 1 1 0 1 0 1 1 1 1 0 0 0\\
\end{array}\right].
\end{eqnarray*}

\begin{figure*}[htbp]
\centering
\includegraphics[width=0.6\columnwidth,height=0.4\linewidth]{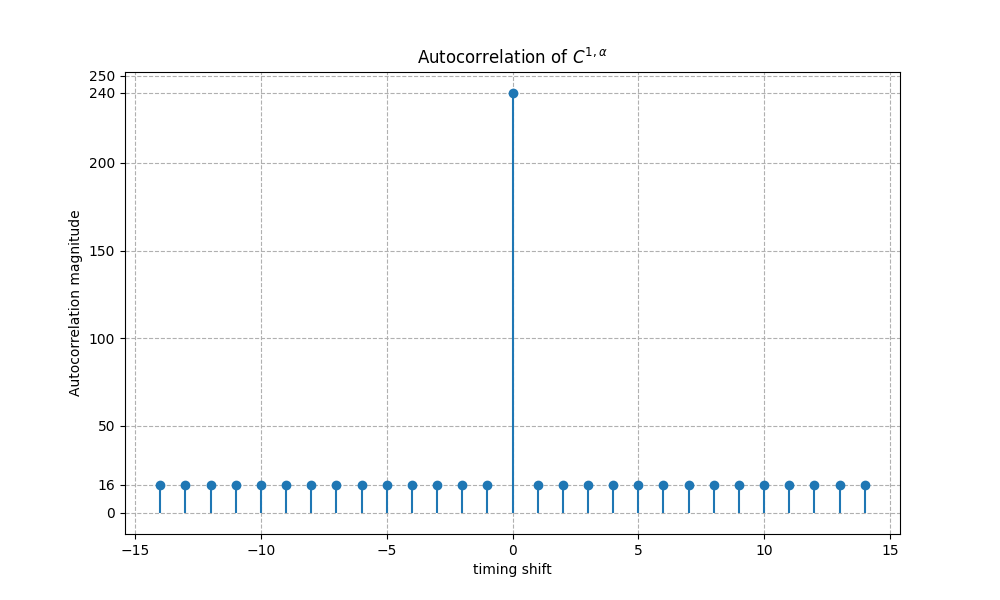}
\caption{The autocorrelation magnitude distribution of $\mathbf{C}^{1,\alpha}$ in Example \ref{example2}}
\label{fig4}
\end{figure*}
\begin{figure*}[htbp]
\centering
\includegraphics[width=0.6\columnwidth,height=0.4\linewidth]{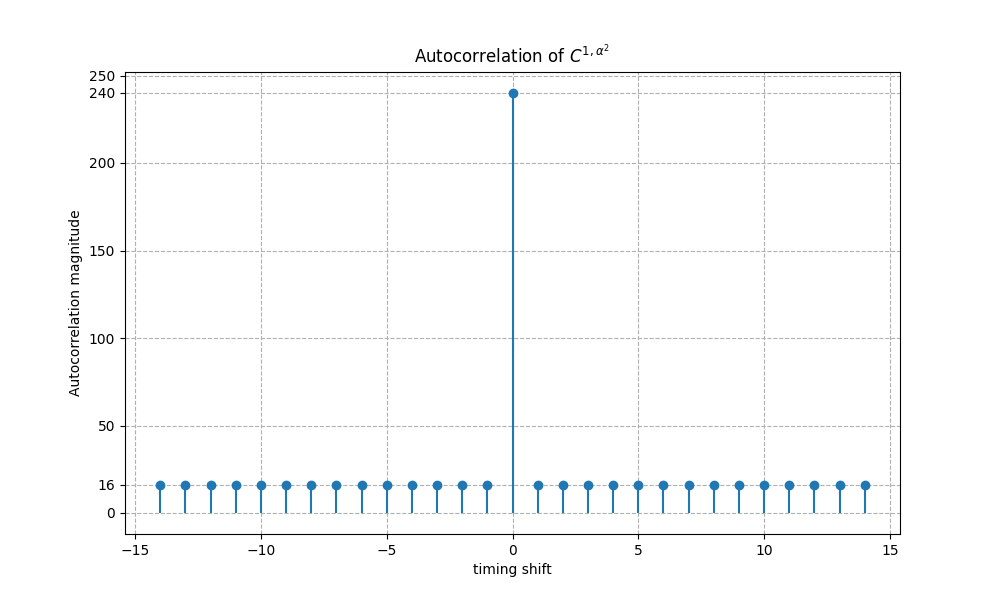}
\caption{The autocorrelation magnitude distribution of $\mathbf{C}^{1,\alpha^2}$ in Example \ref{example2}}
\label{fig5}
\end{figure*}
\begin{figure*}[htbp]
\centering
\includegraphics[width=0.6\columnwidth,height=0.4\linewidth]{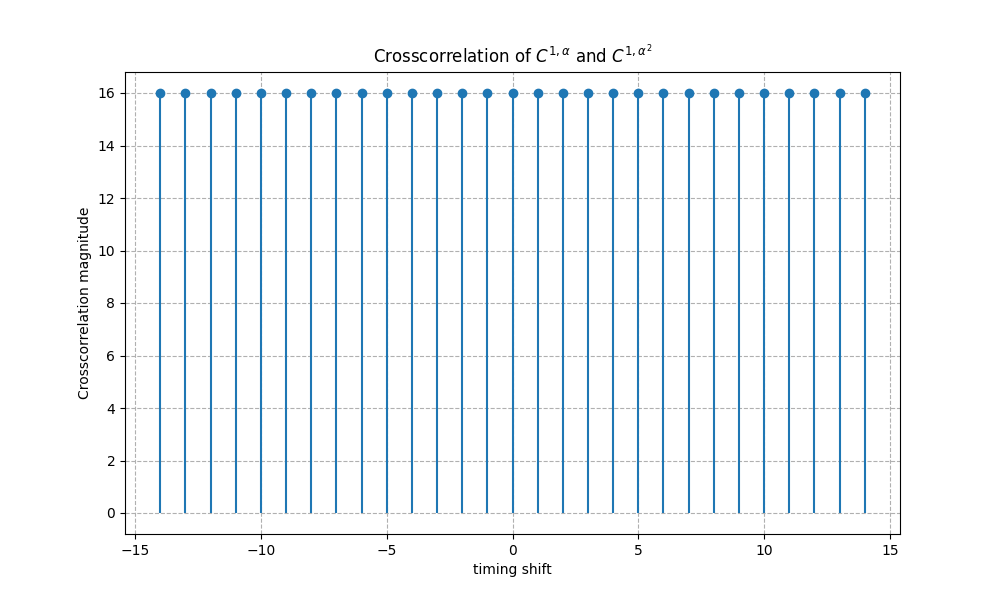}
\caption{The correlation magnitude distribution of $\mathbf{C}^{1,\alpha}$ and $\mathbf{C}^{1,\alpha^2}$ in Example \ref{example2}}
\label{fig6}
\end{figure*}

\end{example}

\subsection{The second construction of periodic QCSSs from a special polynomial }\label{secB}
In this subsection, we will use a special polynomial to construct periodic QCSSs with large set size.

Let $q=p^{n}$ with a prime $p$ and  a positive integer $n$. Let $\alpha$ be a primitive element of $\gf_q$.  Assume that $d_0, d_1,\cdots, d_{q-1}$ are all the elements of $\gf_q$, i.e.  $\gf_q=\{d_0, d_1,\cdots, d_{q-1}\}$. Let $f(x)=x^{p^k+1}+ax^{p^k}+bx+c\in \gf_q[x]$, where $a, b, c\in\gf_q$ and $\gcd(k,n)=1$.
Let $\chi_{1}$ be the canonical additive character of $\gf_q$. We define a constituent sequence $\mathbf{f}_{l}^{a,b,c}$ of period $q-1$ by
\begin{eqnarray}
\mathbf{f}_{l}^{a,b,c}=\left(\mathbf{f}_{l}^{a,b,c}(t)\right)_{t=0}^{q-2}, ~~ \mbox{where }\mathbf{f}_{l}^{a,b,c}(t)=\chi_{1}(\alpha^{t}f(d_l)) \mbox{ for } 0\leq l\leq q-1.
\end{eqnarray}
Then we obtain a two-dimensional $q\times(q-1)$ matrix
\begin{eqnarray*}
\mathbf{F}^{a,b,c}=\left[
\begin{array}{cccc}
\mathbf{f}_{0}^{a,b,c}\\
\mathbf{f}_{1}^{a,b,c}\\
\vdots,\\
\mathbf{f}_{q-1}^{a,b,c}
\end{array}\right].
\end{eqnarray*}
Such matrices form a complementary sequence set 
\begin{eqnarray}\label{eq-c}
\mathcal{F}=\left\{\mathbf{F}^{a,b,c}: a\in\gf_q, b\in \gf_q,c\in \gf_q\right\}.
\end{eqnarray}

\begin{theorem}\label{444} 
Let $p=2$ and $q=2^n$ for a positive integer $n \geq 3$.
Let $f(x)=x^{p^k+1}+ax^{p^k}+bx+c\in \gf_q[x]$, where $a, b, c\in\gf_q$ and $\gcd(k,n)=1$.
Let $\mathcal{F}$ be the complementary sequence set defined in (\ref{eq-c}). Then $\mathcal{F}$ is a periodic $(q^3, q, q-1, 2q)$-QCSS with alphabet size $2$ which  is asymptotically near-optimal with respect to the correlation lower bound in (\ref{wlech}).
\end{theorem}
\begin{proof}
Let $p=2$, $q=2^n$ and $f(x)=x^{p^k+1}+ax^{p^k}+bx+c\in \gf_q[x]$, where $a, b, c\in\gf_q$ and $\gcd(k,n)=1$.
For any two  complementary sequences $\mathbf{F}^{a_{1},b_{1},c_{1}}$, $\mathbf{F}^{a_{2},b_{2},c_{2}}$ in $\mathcal{F}$ and $\tau\in[0,q-2]$, we have
\begin{eqnarray*}
\nonumber &
&R_{\mathbf{F}^{a_{1},b_{1},c_{1}},\mathbf{F}^{a_{2},b_{2},c_{2}}}(\tau)\\
\nonumber 
&=&\sum_{l=0}^{q-1}R_{\mathbf{f}_{l}^{a_{1},b_{1},c_{1}},\mathbf{f}_{l}^{a_{2},b_{2},c_{2}}}(\tau)
\\
\nonumber 
&=&\sum_{l=0}^{q-1}\sum_{t=0}^{q-2}\chi_{1}(\alpha^{t}(d_l^{p^k+1}+a_{1}d_{l}^{p^k}+b_{1}d_{l}+c_1))\overline{\chi_{1}}(\alpha^{t+\tau}(d_l^{p^k+1}+a_{2}d_{l}^{p^k}+b_2d_l+c_2))
\\
&=&\sum_{l=0}^{q-1}\sum_{t=0}^{q-2}\chi_{1}\left((1-\alpha^{\tau})\alpha^{t}d_{l}^{p^k+1}
+(a_1-a_2\alpha^{\tau})\alpha^{t}d_{l}^{p^k}+(b_1-b_2\alpha^{\tau})\alpha^{t}d_l+(c_1-c_2\alpha^{\tau})\alpha^{t}\right),
\end{eqnarray*}
where $a_1,b_1,c_1,a_2,b_2,c_2\in \gf_q$.

We divide into the following cases to determine the value distribution of $R_{\mathbf{F}^{a_{1},b_{1},c_{1}},\mathbf{F}^{a_{2},b_{2},c_{2}}}(\tau)$.

{Case 1}: If $\tau=0$, then 
\begin{eqnarray}\label{eqn-R}
\nonumber &
&R_{\mathbf{F}^{a_{1},b_{1},c_{1}},\mathbf{F}^{a_{2},b_{2},c_{2}}}(\tau)\\
\nonumber 
\nonumber &=&\sum_{l=0}^{q-1}\sum_{t=0}^{q-2}\chi_{1}\left((a_1-a_2)\alpha^{t}d_{l}^{p^k}+(b_1-b_2)\alpha^{t}d_l+(c_1-c_2)\alpha^{t}\right)\\
&=&\sum_{x\in \gf_q}\sum_{y\in \gf_q^*}\chi_{1}\left((a_1-a_2)yx^{p^k}+(b_1-b_2)yx+(c_1-c_2)y\right).
\end{eqnarray}
{Subcase 1.1}: If $a_1-a_2=0$ in Equation (\ref{eqn-R}), by the orthogonal relation of the additive characters, we have
\begin{eqnarray*}
\nonumber &
&R_{\mathbf{F}^{a_{1},b_{1},c_{1}},\mathbf{F}^{a_{2},b_{2},c_{2}}}(\tau)\\
\nonumber 
&=&\sum_{x\in \gf_q}\sum_{y\in \gf_q^*}\chi_{1}\left((b_1-b_2)yx+(c_1-c_2)y\right)\\
&=&\left\{
\begin{array}{ll}
-q & \text{if} ~b_1=b_2 ~\text{and} ~c_1 \neq c_2,\\
0 & \text{if} ~b_1 \neq b_2 ~\text{and} ~c_1 = c_2,\\
0 &  \text{if}~ b_1 \neq b_2 ~\text{and} ~c_1 \neq c_2.
\end{array}\right.
\end{eqnarray*}
Subcase 1.2: If $a_1-a_2 \neq0$  in Equation (\ref{eqn-R}), then
\begin{eqnarray*}
\nonumber &
&R_{\mathbf{F}^{a_{1},b_{1},c_{1}},\mathbf{F}^{a_{2},b_{2},c_{2}}}(\tau)\\
\nonumber 
&=&\sum_{x\in \gf_q}\sum_{y\in \gf_q^*}\chi_{1}\left((a_1-a_2)yx^{p^k}+(b_1-b_2)yx+(c_1-c_2)y\right)\\
&=&\sum_{x\in \gf_q}\sum_{y\in \gf_q^*}\chi_{1}\left(\left((a_1-a_2)x^{p^k}+(b_1-b_2)x+(c_1-c_2)\right)y\right).
\end{eqnarray*}
Now we consider the zeros of polynomial $g(x)=(a_1-a_2)x^{p^k}+(b_1-b_2)x+c_1-c_2$ in $\gf_{2^n}$ with $a_1-a_2 \neq0$, where $p=2$ and $\gcd(n,k)=1$.
Let $N_g$ denote the number of roots of $g(x)=0$ in $\gf_{q}$.
According to Lemma \ref{1112}, we know that $N_g\in\{0, 1, 2\}$ if $p=2$ and $\gcd(n,k)=1$.
By the orthogonal relation of the additive characters, we have
\begin{eqnarray}
\nonumber &
&R_{\mathbf{F}^{a_{1},b_{1},c_{1}},\mathbf{F}^{a_{2},b_{2},c_{2}}}(\tau)\\
\nonumber 
&=&\left\{
\begin{array}{ll}
-q & \text{if} ~N_g=0,\\
0 & \text{if} ~N_g=1,\\
q &  \text{if}~ N_g=2.
\end{array}\right.
\end{eqnarray}

{Case 2}: If $\tau\neq0$, then $1-\alpha^{\tau}\neq0$. Thus we have
\begin{eqnarray*}
\nonumber &
&R_{\mathbf{F}^{a_{1},b_{1},c_{1}},\mathbf{F}^{a_{2},b_{2},c_{2}}}(\tau)\\
\nonumber 
&=&\sum_{x\in \gf_q}\sum_{y\in \gf_q^*}\chi_{1}\left((1-\alpha^{\tau})yx^{p^k+1}
+(a_1-a_2\alpha^{\tau})yx^{p^k}+(b_1-b_2\alpha^{\tau})yx+(c_1-c_2\alpha^{\tau})y\right)\\
&=&\sum_{x\in \gf_q}\sum_{y\in \gf_q^*}\chi_{1}\left(\left((1-\alpha^{\tau})x^{p^k+1}
+(a_1-a_2\alpha^{\tau})x^{p^k}+(b_1-b_2\alpha^{\tau})x+(c_1-c_2\alpha^{\tau})\right)y\right).
\end{eqnarray*}
Now we consider the polynomial $f(x)=(1-\alpha^{\tau})x^{p^k+1}+(a_1-a_2\alpha^{\tau})x^{p^k}+(b_1-b_2\alpha^{\tau})x+c_1-c_2\alpha^{\tau}$ in $\gf_{2^n}$ with $1-\alpha^{\tau}\neq0$, where $p=2$ and $\gcd(n,k)=1$. Let $N_f$ denote the number of zeros of $f$ in $\gf_q$. 
 According to Corollary \ref{cor}, we know that $N_f\in\{0, 1, 2, 3\}$. By the orthogonal relation of additive characters, we then have
\begin{eqnarray}
\nonumber &
&R_{\mathbf{F}^{a_{1},b_{1},c_{1}},\mathbf{F}^{a_{2},b_{2},c_{2}}}(\tau)\\
\nonumber 
&=&\left\{
\begin{array}{ll}
-q & \text{if} ~N_f=0,\\
0 & \text{if} ~N_f=1,\\
q &  \text{if}~ N_f=2,\\
2q & \text{if}~N_f=3.
\end{array}\right.
\end{eqnarray}
When $(a_1,b_1,c_1,a_2,b_2,c_2)$ runs through $\gf_q^6$, each of the coefficients $a_1-a_2\alpha^{\tau}$, $b_1-b_2\alpha^{\tau}$, $a_1-a_2\alpha^{\tau}$ of $f(x)$ also runs through $\gf_q$. By Corollary \ref{cor}, there exists $(a_1,b_1,c_1,a_2,b_2,c_2)\in \gf_q^6$ such that $N_f=3$.

Based on all these cases discussed above, the maximum periodic correlation magnitude of $\mathcal{F}$ is $2q$.
Thus $\mathcal{F}$ is a periodic $(q^3, q, q-1, 2q)$-QCSS. According to the correlation lower bound in (\ref{wlech}), we have
\begin{eqnarray*}
\vartheta_{\opt}=q(q-1)\sqrt{\frac{q^2-1}{q^3(q-1)-1}}=q\sqrt{\frac{(q-1)^{3}(q+1)}{q^3(q-1)-1}}.
\end{eqnarray*}
It is easy to see that 
\begin{eqnarray*}
\lim_{q\rightarrow+\infty}\frac{\vartheta_{\max}}{\vartheta_{\opt}}=\lim_{q\rightarrow+\infty}\frac{2q}{q\sqrt{\frac{(q-1)^{3}(q+1)}{q^3(q-1)-1)}}}=2.
\end{eqnarray*}
Then $\mathcal{F}$ is asymptotically near-optimal with respect to the correlation lower bound in (\ref{wlech}).
This completes the proof of this theorem.\\
\end{proof}

\begin{example}\label{example3}
Let $p=2$ and $n=3$. Then the parameters of the QCSS $\mathcal{F}$ constructed in Theorem \ref{444} are $(512, 8, 7, 16)$ and its alphabet is given by $\{(-1)^i: i\in[0,1]\}$. By Magma program, we verify that the periodic correlation function $R_{\mathbf{F}^{a_{1},b_{1},c_{1}},\mathbf{F}^{a_{2},b_{2},c_{1}}}(\tau)$ is equal to $0$, $8$ or $16$ for any $a,b,c\in \gf_q$ and $0\leq\tau\leq 6$ except the trivial case that $a_1=a_2, b_1=b_2,c_1=c_2$ and $\tau=0$. For instance, the matrices $\mathbf{F}^{1,\alpha,\alpha}$ and $\mathbf{F}^{1,\alpha,\alpha^2}$ in $\mathcal{F}$ are presented as follows, where each entry stands for a power of $\zeta_2=-1$.  By Python program, we show the autocorrelation magnitude distribution of $\mathbf{F}^{1,\alpha,\alpha}$ in Fig. \ref{fig7},  the autocorrelation magnitude  distribution of $\mathbf{F}^{1,\alpha,\alpha^2}$ in Fig. \ref{fig8}, and the correlation magnitude distribution of $\mathbf{F}^{1,\alpha,\alpha}$ and $\mathbf{F}^{1,\alpha,\alpha^2}$ in Fig. \ref{fig9}, respectively.  
\end{example}
\begin{eqnarray*}\label{matrix2}
\mathbf{F}^{1,\alpha,\alpha}=\left[
\begin{array}{cccc}
0 0 0 0 0 0 0\\
0 0 1 0 1 1 1\\
0 1 0 1 1 1 0\\
1 0 0 1 0 1 1\\
0 0 0 0 0 0 0\\
1 1 1 0 0 1 0\\
0 1 0 1 1 1 0\\
0 1 0 1 1 1 0\\
\end{array}\right],\
\mathbf{F}^{1,\alpha,\alpha^2}=\left[
\begin{array}{cccc}
1 1 1 0 0 1 0\\
1 1 0 0 1 0 1\\
1 0 1 1 1 0 0\\
0 1 1 1 0 0 1\\
1 1 1 0 0 1 0\\
0 0 0 0 0 0 0\\
1 0 1 1 1 0 0\\
1 0 1 1 1 0 0\\
\end{array}\right].
\end{eqnarray*}
\begin{figure*}[htbp]
\centering
\includegraphics[width=0.6\columnwidth,height=0.4\linewidth]{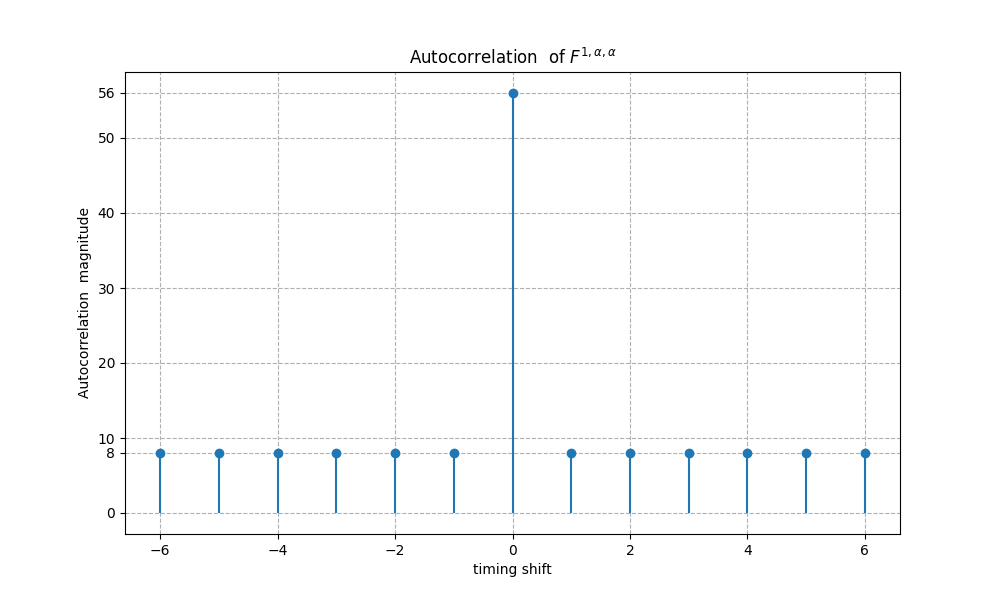}
\caption{The autocorrelation magnitude distribution of $\mathbf{F}^{1,\alpha,\alpha}$ in Example \ref{example3}}
\label{fig7}
\end{figure*}
\begin{figure*}[htbp]
\centering
\includegraphics[width=0.6\columnwidth,height=0.4\linewidth]{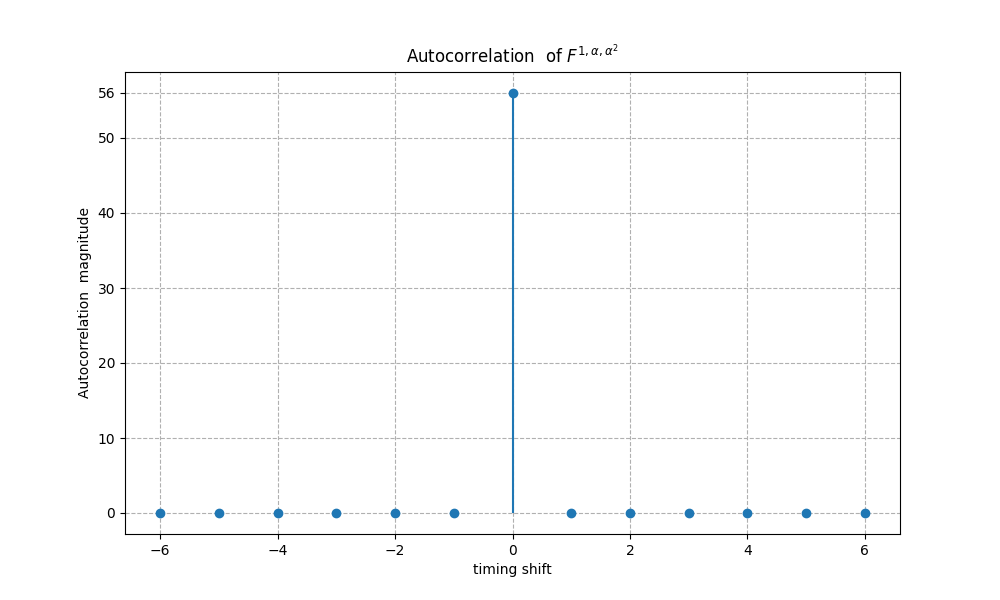}
\caption{The autocorrelation  magnitude distribution of $\mathbf{F}^{1,\alpha,\alpha^2}$ in Example \ref{example3}}
\label{fig8}
\end{figure*}
\begin{figure*}[htbp]
\centering
\includegraphics[width=0.6\columnwidth,height=0.4\linewidth]{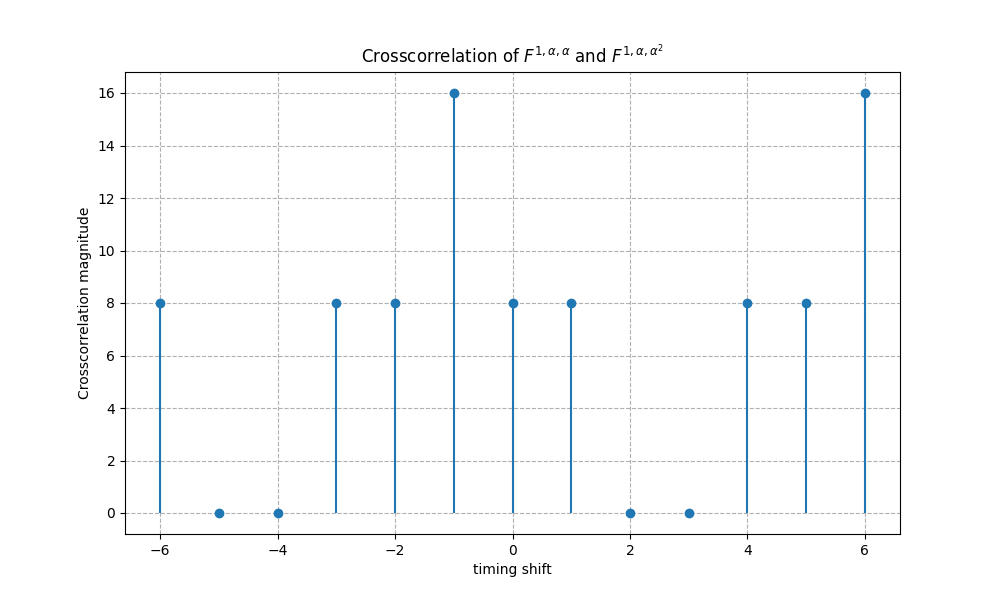}
\caption{The correlation magnitude distribution of $\mathbf{F}^{1,\alpha,\alpha}$ and $\mathbf{F}^{1,\alpha,\alpha^2}$ in Example \ref{example3}}
\label{fig9}
\end{figure*}

\subsection{The third construction of periodic QCSSs from a cubic polynomial}
In this subsection, we use a cubic polynomial to construct periodic QCSSs with large set size.
Let $q=p^n$ with an odd prime $p$ and a positive integer $n$. Let $\alpha$ be a primitive element of $\gf_q$. Assume that $d_0, d_1,\ldots, d_{q-1}$ are all the elements of $\gf_q$, i.e. $\gf_q=\{d_0, d_1,\ldots, d_{q-1}\}$. Let $f(x)=x^3+ax^2+bx+c\in \gf_q[x]$, where $a,b,c\in \gf_q$. Let $\chi_1$ be the canonical additive character of $\gf_q$. We define a constituent sequence $\mathbf{h}_{l}^{a,b,c}$ of period $q-1$ by 
\begin{eqnarray*}
\mathbf{h}_{l}^{a,b,c}=\left(\mathbf{h}_{l}^{a,b,c}(t)\right)_{t=0}^{q-2}, ~~ \text{where}\ \mathbf{h}_{l}^{a,b,c}(t)=\chi_{1}\left(\alpha^{t}f(d_l)\right)\mbox{ for }0\leq l\leq q-1.
\end{eqnarray*}
Then we obtain a two-dimensional $q\times(q-1)$ matrix 
\begin{eqnarray*}
\mathbf{H}^{a,b,c}=\left[
\begin{array}{cccc}
\mathbf{h}_{0}^{a,b,c}\\
\mathbf{h}_{l}^{a,b,c}\\
\vdots\\
\mathbf{h}_{q-1}^{a,b,c}\\
\end{array}\right].
\end{eqnarray*}
Such matrices 
constitute a complementary sequence set
\begin{eqnarray}\label{eq-c3}
\mathcal{H}=\left\{\mathbf{H}^{a,b,c}:a\in \gf_q, b\in \gf_q, c\in \gf_q\right\}.
\end{eqnarray}
\begin{theorem}\label{333}
Let $q=p^n$ with an odd prime $p$ and a positive integer $n$. Let $f(x)=x^3+ax^2+bx+c\in \gf_q[x]$, where $a,b,c\in \gf_q$. Let $\mathcal{H}$ be the complementary sequence set defined in (\ref{eq-c3}). Then $\mathcal{H}$ is a periodic $(q^3, q, q-1, 2q)$-QCSS with alphabet size $p$ which is asymptotically near-optimal with respect to the correlation lower bound in(\ref{wlech}).
\end{theorem}
\begin{proof}
For any two  complementary sequences $\mathbf{H}^{a_{1},b_{1},c_{1}}$, $\mathbf{H}^{a_{2},b_{2},c_{2}}$ of $\mathcal{H}$ and $\tau\in[0,q-2]$, we have
\begin{eqnarray*}
\nonumber &
&R_{\mathbf{H}^{a_{1},b_{1},c_{1}},\mathbf{H}^{a_{2},b_{2},c_{2}}}(\tau)\\
\nonumber 
&=&\sum_{l=0}^{q-1}R_{\mathbf{h}_{l}^{a_{1},b_{1},c_{1}},\mathbf{h}_{l}^{a_{2},b_{2},c_{2}}}(\tau)
\\
\nonumber 
&=&\sum_{l=0}^{q-1}\sum_{t=0}^{q-2}\chi_{1}\left(\alpha^{t}(d_l^{3}+a_{1}d_{l}^{2}+b_{1}d_{l}+c_1)\right)\overline{\chi_{1}}\left(\alpha^{t+\tau}(d_l^{3}+a_{2}d_{l}^{2}+b_2d_l+c_2)\right)
\\
&=&\sum_{l=0}^{q-1}\sum_{t=0}^{q-2}\chi_{1}\left((1-\alpha^{\tau})\alpha^{t}d_{l}^{3}+(a_1-a_2\alpha^{\tau})\alpha^{t}d_{l}^{2}
+(b_1-b_2\alpha^{\tau})\alpha^{t}d_l+(c_1-c_2\alpha^{\tau})\alpha^{t}\right),
\end{eqnarray*}
where $a_1, a_2, b_1, b_2, c_1, c_2\in \gf_q$.
We consider the following cases to determine its value distribution.

{Case 1}: If $\tau=0$, then 
\begin{eqnarray}\label{eqq1}
\nonumber &
&R_{\mathbf{H}^{a_{1},b_{1},c_{1}},\mathbf{H}^{a_{2},b_{2},c_{2}}}(\tau)\\
\nonumber 
&=&\sum_{l=0}^{q-1}\sum_{t=0}^{q-2}\chi_{1}\left((a_1-a_2)\alpha^td_l^{2}+(b_1-b_2)\alpha^td_l+(c_1-c_2)\alpha^t\right)
\\
&=&\sum_{x\in\gf_q}\sum_{y\in\gf_q^{*}}\chi_{1}\left((a_1-a_2)yx^{2}+(b_1-b_2)yx+(c_1-c_2)y\right).
\end{eqnarray}
Subcase 1.1: If $a_1-a_2=0$ in Equation (\ref{eqq1}), by the orthogonal relation of the additive characters, we have
\begin{eqnarray*}
\nonumber &
&R_{\mathbf{H}^{a_{1},b_{1},c_{1}},\mathbf{H}^{a_{2},b_{2},c_{2}}}(\tau)\\
\nonumber 
&=&\sum_{x\in\gf_q}\sum_{y\in\gf_q^{*}}\chi_{1}((b_1-b_2)yx+(c_1-c_2)y)\\
&=&\left\{
\begin{array}{ll}
-q & \mbox{if $b_1=b_2$ and $c_1 \neq c_2$},\\
0 & \mbox{if $b_1\neq b_2$ and $c_1 = c_2$},\\
0 &  \mbox{if $b_1\neq b_2$ and $c_1 \neq c_2$}.
\end{array}\right.
\end{eqnarray*}
Subcase 1.2: If $a_1-a_2\neq0$, then
\begin{eqnarray*}
\nonumber &
&R_{\mathbf{H}^{a_{1},b_{1},c_{1}},\mathbf{H}^{a_{2},b_{2},c_{2}}}(\tau)\\
\nonumber 
&=&\sum_{x\in\gf_q}\sum_{y\in\gf_q^{*}}\chi_{1}((a_1-a_2)yx^{2}+(b_1-b_2)yx+(c_1-c_2)y)\\
&=&\sum_{x\in\gf_q}\sum_{y\in\gf_q^{*}}\chi_{1}\left(\left((a_1-a_2)x^{2}+(b_1-b_2)x+c_1-c_2\right)y\right).
\end{eqnarray*}
Now we consider the zeros of the polynomial $g(x)=(a_1-a_2)x^{2}+(b_1-b_2)x+c_1-c_2$ in $\gf_q$ with $a_1-a_2\neq0$, where $p$ is an odd prime. Let $N_g$ denote the number of roots of $g(x)=0$ in $\gf_q$. It is easy to know that $N_g\in \{0, 1, 2\}$. By the orthogonal relation of the additive characters, we have
\begin{eqnarray}
\nonumber &
&R_{\mathbf{H}^{a_{1},b_{1},c_{1}},\mathbf{H}^{a_{2},b_{2},c_{2}}}(\tau)\\
\nonumber 
&=&\left\{
\begin{array}{ll}
-q & \mbox{if $N_f=0$},\\
0 & \mbox{if $N_f=1$},\\
q &  \mbox{if $N_f=2$}.
\end{array}\right.
\end{eqnarray}

{Case 2}: If $\tau\neq0$, then $1-\alpha^{\tau}\neq0$. Thus we have
\begin{eqnarray*}
\nonumber &
&R_{\mathbf{H}^{a_{1},b_{1},c_{1}},\mathbf{H}^{a_{2},b_{2},c_{2}}}(\tau)\\
\nonumber 
&=&\sum_{x\in\gf_q}\sum_{y\in\gf_q^{*}}\chi_{1}\left((1-\alpha^{\tau})yx^{3}+(a_1-a_2\alpha^{\tau})yx^{2}+(b_1-b_2\alpha^{\tau})yx+(c_1-c_2\alpha^{\tau})y\right)\\
&=&\sum_{x\in\gf_q}\sum_{y\in\gf_q^{*}}\chi_{1}\left(\left((1-\alpha^{\tau})x^{3}+(a_1-a_2\alpha^{\tau})x^{2}+(b_1-b_2\alpha^{\tau})x+c_1-c_2\alpha^{\tau}\right)y\right).
\end{eqnarray*}
Now we consider the zeros of the polynomial $f(x)=(1-\alpha^{\tau})x^{3}+(a_1-a_2\alpha^{\tau})x^{2}+(b_1-b_2\alpha^{\tau})x+c_1-c_2\alpha^{\tau}$ in $\gf_q$ with $1-\alpha^{\tau}\neq0$, where $p$ is an odd prime. Let $N_f$ denote the number of zeros of $f$ in $\gf_q$. It is known that $N_f \in \{0, 1, 2, 3\}$.
Note that the each of the coefficients $a_1-a_2\alpha^{\tau}$, $b_1-b_2\alpha^{\tau}$ and $c_1-c_2\alpha^{\tau}$ 
runs over $\gf_q$ when all of $a_1,a_1,b_1,b_2,c_1,c_2$ run over $\gf_q$.
Note that there exists a polynomial $f(x)=(1-\alpha^{\tau})(x-x_1)(x-x_2)(x-x_3)$ for pairwise distinct elements $x_1,x_2,x_3\in \gf_q$ which has three zeros in $\gf_q$. 
By the orthogonal relation of additive characters, we then have
\begin{eqnarray}
\nonumber &
&R_{\mathbf{H}^{a_{1},b_{1},c_{1}},\mathbf{H}^{a_{2},b_{2},c_{2}}}(\tau)\\
\nonumber 
&=&\left\{
\begin{array}{ll}
-q & \mbox{if $N_f=0$},\\
0 & \mbox{if $N_f=1$},\\
q &  \mbox{if $N_f=2$},\\
2q & \mbox{if $N_f=3$}.
\end{array}\right.
\end{eqnarray}
Based on all the cases discussed above, the maximum periodic correlation magnitude of $\mathcal{H}$ is $2q$. Thus $\mathcal{H}$ is a  periodic $(q^3, q, q-1,2q)$-QCSS. According to the correlation lower bound in (\ref{wlech}), we have 
\begin{eqnarray*}
\vartheta_{\opt}=q(q-1)\sqrt{\frac{q^2-1}{q^3(q-1)-1}}=q\sqrt{\frac{(q-1)^{3}(q+1)}{q^3(q-1)-1}}.
\end{eqnarray*}
It is easy to see that 
\begin{eqnarray*}
\lim_{q\rightarrow+\infty}\frac{\vartheta_{\max}}{\vartheta_{\opt}}= \lim_{q\rightarrow+\infty}\frac{2q}{q\sqrt{\frac{(q-1)^{3}(q+1)}{q^3(q-1)-1)}}}=2.
\end{eqnarray*}
Then $\mathcal{H}$ is asymptotically near-optimal with respect to the correlation lower bound in (\ref{wlech}). This completes the proof of this theorem.\\
\end{proof}
\begin{example}\label{example6}
Let $p=3$ and $n=2$. Then the parameters of the QCSS $\mathcal{H}$ constructed in Theorem \ref{333} are $(729, 9, 8, 18)$ and its alphabet is given by $\{e^{2\pi\sqrt{-1}i/3}: i\in[0,2]\}$. By Magma program, we verify that the periodic correlation function $R_{\mathbf{H}^{a_{1},b_{1},c_{1}},\mathbf{H}^{a_{2},b_{2},c_{1}}}(\tau)$ is equal to $0$, $9$ or $18$ for any $a,b,c\in \gf_q$ and $0\leq\tau\leq7$ except the trivial case that $a_1=a_2, b_1=b_2,c_1=c_2$ and $\tau=0$. For instance, the matrices $\mathbf{H}^{1,\alpha,\alpha}$ and $\mathbf{H}^{1,\alpha,\alpha^2}$ in $\mathcal{H}$ are presented as follows, where each entry stands for a power of $\zeta_3=e^{2\pi\sqrt{-1}/3}$.
By Python program, we show the autocorrelation magnitude distribution of $\mathbf{H}^{1,\alpha,\alpha}$ in Fig. $10$,  the autocorrelation magnitude
distribution of $\mathbf{H}^{1,\alpha,\alpha^2}$ in Fig. $11$, and the correlation magnitude distribution of $\mathbf{H}^{1,\alpha,\alpha}$ and $\mathbf{H}^{1,\alpha,\alpha^2}$ in Fig. $12$. We remark that the upper bound of $\vartheta_{\max}$ in Theorem \ref{333} is tight in this case. 
\begin{eqnarray*}\label{matrix6}
\mathbf{h}^{1,\alpha,\alpha}=\left[
\begin{array}{cccc}
2 2 1 0 1 1 2 0\\
0 2 2 1 0 1 1 2\\
2 2 1 0 1 1 2 0\\
1 1 2 0 2 2 1 0\\
0 0 0 0 0 0 0 0\\
2 2 1 0 1 1 2 0\\
2 0 2 2 1 0 1 1\\
0 2 2 1 0 1 1 2\\
0 1 1 2 0 2 2 1\\
\end{array}\right],\
\mathbf{h}^{1,\alpha,\alpha^2}=\left[
\begin{array}{cccc}
0 2 2 1 0 1 1 2\\
1 2 0 2 2 1 0 1\\
0 2 2 1 0 1 1 2\\
2 1 0 1 1 2 0 2\\
1 0 1 1 2 0 2 2\\
0 2 2 1 0 1 1 2\\
0 0 0 0 0 0 0 0\\
1 2 0 2 2 1 0 1\\
1 1 2 0 2 2 1 0\\
\end{array}\right].
\end{eqnarray*}
\begin{figure*}[htbp]
\centering
\includegraphics[width=0.6\columnwidth,height=0.4\linewidth]{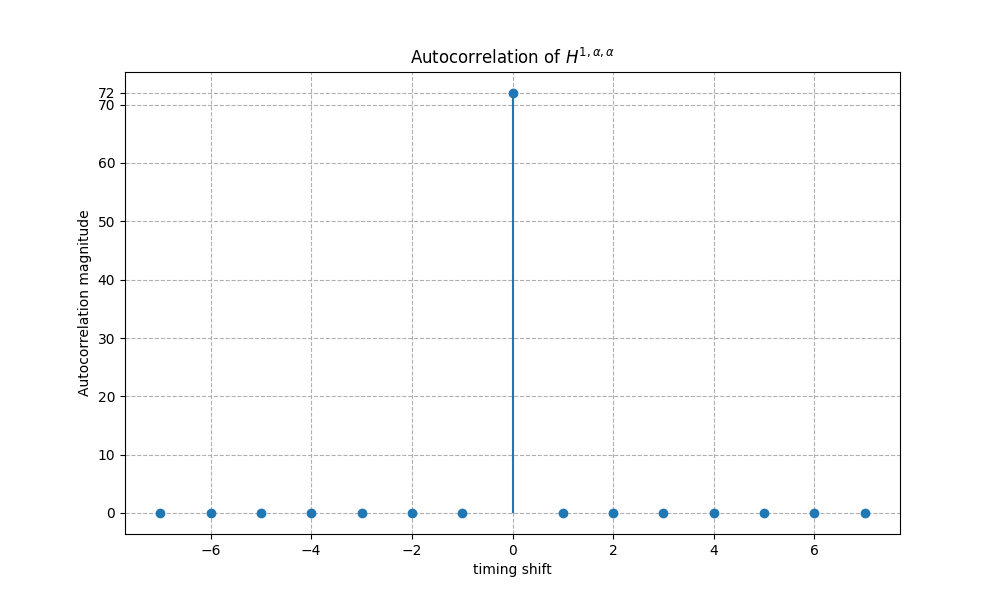}
\caption{The autocorrelation magnitude distribution of $\mathbf{H}^{1,\alpha,\alpha}$ in Example \ref{example6} }
\label{fig9.1}
\end{figure*}
\begin{figure*}[htbp]
\centering
\includegraphics[width=0.6\columnwidth,height=0.4\linewidth]{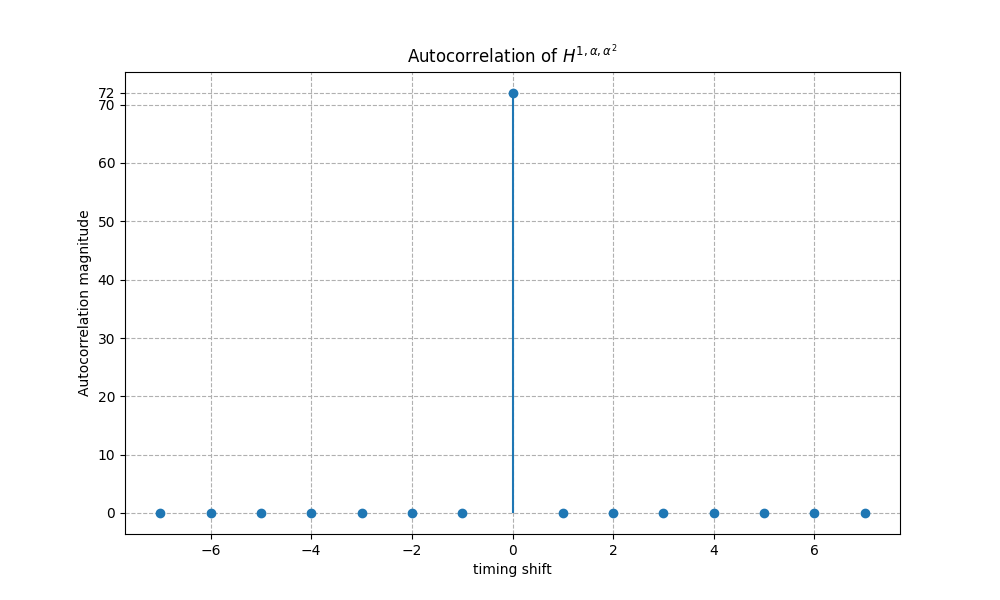}
\caption{The autocorrelation magnitude distribution of $\mathbf{H}^{1,\alpha,\alpha^2}$ in Example \ref{example6}}
\label{fig9.2}
\end{figure*}
\begin{figure*}[htbp]
\centering
\includegraphics[width=0.6\columnwidth,height=0.4\linewidth]{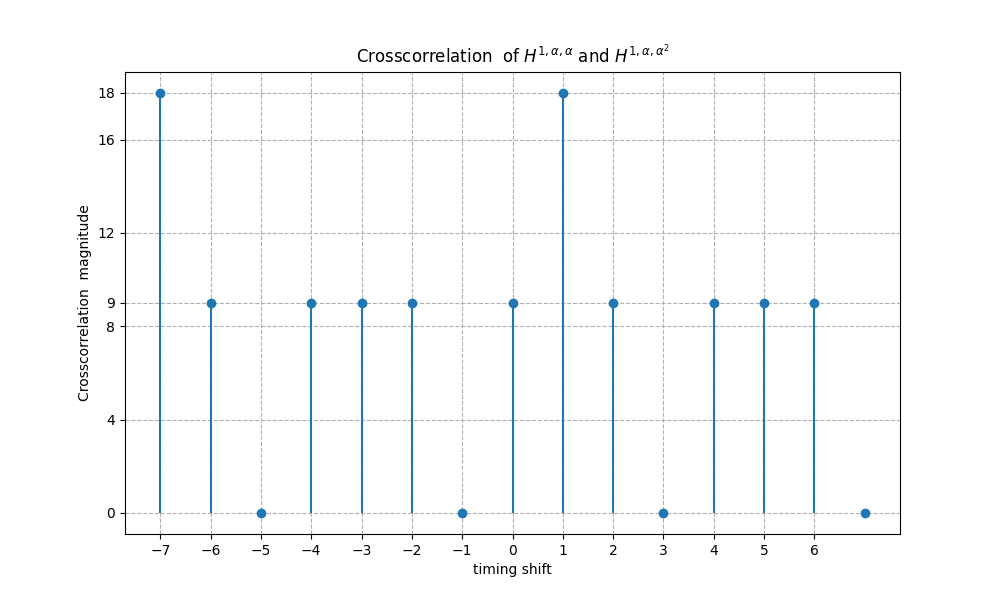}
\caption{ The correlation magnitude  distribution of $\mathbf{H}^{1,\alpha,\alpha}$ and $\mathbf{H}^{1,\alpha,\alpha^2}$ in Example \ref{example6}}
\label{fig9.3}
\end{figure*}
\end{example}

\subsection{The fourth construction of periodic QCSSs from Gaussian sums}

Let $q=p^{n}$ with a prime $p$ and  a positive integer $n$. Let $\alpha$ be a primitive element of $\gf_q$.  Assume that $d_0, d_1,\cdots, d_{q-2}$ are all the elements of $\gf_q^*$, i.e.  $\gf_q^*=\{d_0, d_1,\cdots, d_{q-2}\}$.
Let $\chi_{b}$, $b\in \gf_q$, and $\varphi_{i} $, $0\leq i \leq q-2$, be  additive character and  multiplicative characters of $\gf_q$ defined in Subsection \ref{subsection}, respectively.

For each $b\in \gf_q$ and $0\leq i \leq q-2$, we define a constituent sequence $\mathbf{s}_{l}^{i,b}$ of period $q-1$ as 
\begin{eqnarray*}
\mathbf{s}_{l}^{i,b}=\left(\mathbf{s}_{l}^{i,b}(t)\right)_{t=0}^{q-2}, ~\text{where}\ \mathbf{s}_{l}^{i,b}=\varphi_{i}(d_l)\chi_{b}(\alpha^{t}d_l)\chi_{1}({\alpha^{t}})~\text{for}\ 0\leq l \leq q-2.
\end{eqnarray*}
Then we obtain a two-dimensional $(q-1)\times(q-1)$ matrix
\begin{eqnarray*}
\mathbf{S}^{i,b}=\left[
\begin{array}{cccc}
\mathbf{s}_{0}^{i,b}\\
\mathbf{s}_{1}^{i,b}\\
\vdots,\\
\mathbf{s}_{q-2}^{i,b}
\end{array}\right].
\end{eqnarray*}
The set 
\begin{eqnarray}\label{eq-c2}
\mathcal{S}=\left\{\mathbf{S}^{i,b}: 0\leq i \leq q-2, b\in \gf_q\right\}
\end{eqnarray}
is a complementary sequence set.
\begin{theorem}\label{555}
Let $q=p^n$ with a prime $p$ and a  positive integer $n$. Let $\mathcal{S}$ be the complementary sequence set
defined in (\ref{eq-c2}). Then $\mathcal{S}$ is a periodic $(q^2-q, q-1, q-1, q)$-QCSS with alphabet size $p(p^n-1)$ which is asymptotically optimal with respect
to the correlation lower bound in (\ref{wlech}).
\end{theorem}
\begin{proof}
For any two  complementary sequences $\mathbf{S}^{i_1,b_{1}}$, $\mathbf{S}^{i_2,b_{2}}$ in $\mathcal{S}$ and $\tau\in[0,q-2]$, we have
\begin{eqnarray*}
\nonumber &
&R_{\mathbf{S}^{i_1,b_{1}},\mathbf{S}^{i_2,b_{2}}}(\tau)\\
\nonumber 
&=&\sum_{l=0}^{q-2}R_{\mathbf{s}_{l}^{i_1,b_{1}},\mathbf{s}_{l}^{i_2,b_{2}}}(\tau)\\
&=&\sum_{l=0}^{q-2}\sum_{t=0}^{q-2}\varphi_{i_1}(d_l)\chi_{b_1}(d_{l}\alpha^{t})\chi_{1}(\alpha^{t}) \overline{\varphi_{i_2}}(d_l)\overline{\chi_{b_2}}(d_{l}\alpha^{t+\tau})\overline{\chi_{1}}(\alpha^{t+\tau})\\
&=&\sum_{l=0}^{q-2}\sum_{t=0}^{q-2}\varphi_{i_1-i_2}(d_l)\chi_{1}\left((b_{1}-b_{2}\alpha^{\tau})\alpha^{t}d_{l}\right)\chi_{1}\left((1-\alpha^{\tau})\alpha^{t}\right),
\end{eqnarray*}
where  $0\leq i_1 \leq q-2, 0\leq i_2 \leq q-2, b_1,b_2\in \gf_q$.
Then we consider the following cases to determine the value distribution of $R_{\mathbf{S}^{i_1,b_{1}},\mathbf{S}^{i_2,b_{2}}}(\tau)$.

{Case 1}: If $\tau =0$, $i_1=i_2$ and $b_{1}\neq b_{2}$, then
\begin{eqnarray*}
\nonumber &
&R_{\mathbf{S}^{i_1,b_{1}},\mathbf{S}^{i_2,b_{2}}}(\tau)\\
&=&\sum_{l=0}^{q-2}\sum_{t=0}^{q-2}\chi_{1}\left((b_{1}-b_{2})\alpha^{t}d_{l}\right)\\
&=&\sum_{l=0}^{q-2}\sum_{x\in \gf_q^*}\chi_{1}\left((b_{1}-b_{2})d_{l}x\right).
\end{eqnarray*}
By the orthogonal relation of the additive characters, we have
\begin{eqnarray*}
R_{\mathbf{S}^{i_1,b_{1}},\mathbf{S}^{i_2,b_{2}}}(\tau)=-(q-1).
\end{eqnarray*}

{Case 2}: If $\tau =0$, $i_1\neq i_2$ and $b_{1}=b_{2}$, then
\begin{eqnarray*}
\nonumber &
&R_{\mathbf{S}^{i_1,b_{1}},\mathbf{S}^{i_2,b_{2}}}(\tau)\\
&=&\sum_{l=0}^{q-2}\sum_{t=0}^{q-2}\varphi_{i_1-i_2}(d_{l})\\
&=&(q-1)\sum_{y\in \gf_q^*}\varphi_{i_1-i_2}(y)=0
\end{eqnarray*}
by the orthogonal relation of the multiplicative characters.

{Case 3}: If $\tau =0$, $i_1\neq i_2$ and $b_{1}\neq b_{2}$, then
\begin{eqnarray*}
\nonumber &
&R_{\mathbf{S}^{i_1,b_{1}},\mathbf{S}^{i_2,b_{2}}}(\tau)\\
&=&\sum_{l=0}^{q-2}\sum_{t=0}^{q-2}\varphi_{i_1-i_2}(d_{l})\chi_{1}((b_{1}-b_{2})\alpha^{t}d_{l})\\
&=&\sum_{y\in \gf_q^*}\varphi_{i_1-i_2}(y)\sum_{x\in \gf_q^*}\chi_{1}((b_{1}-b_{2})yx)\\
&=&-\sum_{y\in \gf_q^*}\varphi_{i_1-i_2}(y)=0
\end{eqnarray*}
by the orthogonal relation of the additive and multiplicative characters.

{Case 4}: If $\tau \neq0$ and $i_1= i_2$, then
\begin{eqnarray*}
\nonumber &
&R_{\mathbf{S}^{i_1,b_{1}},\mathbf{S}^{i_2,b_{2}}}(\tau)\\
&=&\sum_{l=0}^{q-2}\sum_{t=0}^{q-2}\chi_{1}((b_{1}-b_{2}\alpha^{\tau})\alpha^{t}d_{l})\chi_{1}((1-\alpha^{\tau})\alpha^{t})\\
&=&\sum_{x\in \gf_q^*}\chi_{1}((1-\alpha^{\tau})x)\sum_{y\in \gf_q^*}\chi_{1}((b_{1}-b_{2}\alpha^{\tau})xy)\\
&=&\left\{
\begin{array}{ll}
(q-1)\sum_{x\in \gf_q^*}\chi_{1}((1-\alpha^{\tau})x) & \text{if}~ b_{1}-b_{2}\alpha^{\tau}=0\\
-\sum_{x\in \gf_q^*}\chi_{1}((1-\alpha^{\tau})x) & \text{if}~ b_{1}-b_{2}\alpha^{\tau}\neq0
\end{array}\right.\\
&=&\left\{
\begin{array}{ll}
-(q-1) & \text{if}~ b_{1}-b_{2}\alpha^{\tau}=0,\\
1 & \text{if}~ b_{1}-b_{2}\alpha^{\tau}\neq0,
\end{array}\right.
\end{eqnarray*}
according to the orthogonal relation of the additive characters.

{Case 5}: If $\tau \neq0$ and $i_1\neq i_2$, then
\begin{eqnarray*}
\nonumber &
&R_{\mathbf{S}^{i_1,b_{1}},\mathbf{S}^{i_2,b_{2}}}(\tau)\\
&=&\sum_{l=0}^{q-2}\sum_{t=0}^{q-2}\varphi_{i_1-i_2}(d_l)\chi_{1}((b_{1}-b_{2}\alpha^{\tau})\alpha^{t}d_{l})\chi_{1}((1-\alpha^{\tau})\alpha^{t})\\
&=&\sum_{x\in \gf_q^*}\sum_{y\in \gf_q^*}\varphi_{i_1-i_2}(y)\chi_{1}((b_{1}-b_{2}\alpha^{\tau})xy)\chi_{1}((1-\alpha^{\tau})x)\\
&=&\left\{
\begin{array}{ll}
\sum\limits_{y\in \gf_q^*}\varphi_{i_1-i_2}(y)\sum\limits_{x\in \gf_q^*}\chi_{1}((1-\alpha^{\tau})x) & \text{if}~ b_{1}-b_{2}\alpha^{\tau}=0\\
\substack{\mbox{$\sum\limits_{x\in \gf_q^*}\chi_{1}((1-\alpha^{\tau})x)\overline{\varphi_{i_1-i_2}}((b_{1}-b_{2}\alpha^{\tau})x)\cdot$}\\
\mbox{$\sum\limits_{y\in \gf_q^*}\varphi_{i_1-i_2}((b_{1}-b_{2}\alpha^{\tau})xy)\chi_{1}((b_{1}-b_{2}\alpha^{\tau})xy)$}} & \text{if}~ b_{1}-b_{2}\alpha^{\tau}\neq0
\end{array}\right.\\
&=&\left\{
\begin{array}{ll}
0 & \text{if}~ b_{1}-b_{2}\alpha^{\tau}=0,\\
\overline{\varphi_{i_1-i_2}}(b_{1}-b_{2}\alpha^{\tau})\varphi_{i_1-i_2}(1-\alpha^{\tau})G(\varphi_{i_1-i_2},\chi_{1}) G(\overline{\varphi_{i_1-i_2}},\chi_{1})& \text{if}~ b_{1}-b_{2}\alpha^{\tau}\neq0,
\end{array}\right.
\end{eqnarray*}
by the orthogonal relation of the additive and multiplicative characters and the definition of Gaussian sums.
It is known that $G(\overline{\varphi_{i_1-i_2}},\chi_{1})=\varphi_{i_1-i_2}(-1)\overline{G(\varphi_{i_1-i_2},\chi_{1})}$.
Then $$G(\varphi_{i_1-i_2},\chi_{1}) G(\overline{\varphi_{i_1-i_2}},\chi_{1})=\varphi_{i_1-i_2}(-1)|G(\varphi,\chi_1)|^2=q\varphi_{i_1-i_2}(-1)$$
by Lemma \ref{quadGuasssum2}. Thus
\begin{eqnarray*}
|R_{\mathbf{C}^{a_{1},b_{1}},\mathbf{C}^{a_{2},b_{2}}}(\tau)|\in \{0,  \pm q\}.
\end{eqnarray*}

Based on the above discussions, we deduce that the maximum periodic correlation magnitude of $\mathcal{S}$ is $q$.
Then $\mathcal{S}$ is a periodic $(q^{2}-q, q-1, q-1, q)$-QCSS. According to the correlation bound in (\ref{wlech}), we have 
\begin{eqnarray*}
\vartheta_{\opt}=\sqrt{\frac{(q-1)^{5}}{q(q-1)^2-1}}.
\end{eqnarray*}
It is obvious that 
\begin{eqnarray*}
\lim_{q\rightarrow+\infty}\frac{\vartheta_{\max}}{\vartheta_{\opt}}=\lim_{q\rightarrow+\infty}\frac{q}{\sqrt{\frac{(q-1)^{5}}{q(q-1)^2-1}}}=1.
\end{eqnarray*}
Then $\mathcal{S}$ is asymptotically optimal with respect to the correlation lower bound in  (\ref{wlech}).
This completes the proof of this theorem.
\end{proof}
\begin{example}\label{example4}
Let $p=2$ and $n=3$. Then the parameters of the QCSS $\mathcal{S}$ constructed in Theorem \ref{555} are $(56, 7, 7, 8)$.
By Magma program, the periodic correlation function $R_{\mathbf{S}^{i_1 ,b_{1}},\mathbf{S}^{i_2,b_{2}}}(\tau)$ is equal to $0$, $1$, $7$ or $8$ for any $b\in \gf_q$, $0\leq i \leq 6$ and $0\leq\tau\leq 6$ except the trivial case that $i_1=i_2, b_1=b_2$ and $\tau=0$. By Python program, we show the autocorrelation magnitude distribution of  $\mathbf{S}^{1,\alpha}$  in Fig. \ref{fig13},  the autocorrelation magnitude  distribution of $\mathbf{S}^{1,\alpha^2}$ in Fig. \ref{fig14}, and the correlation magnitude distribution of $\mathbf{S}^{1,\alpha}$ and $\mathbf{S}^{1,\alpha^2}$ in Fig. \ref{fig15}, respectively. These correlation magnitude distributions coincide with the results in Theorem \ref{555}.

\begin{figure*}[htbp]
\centering
\includegraphics[width=0.6\columnwidth,height=0.4\linewidth]{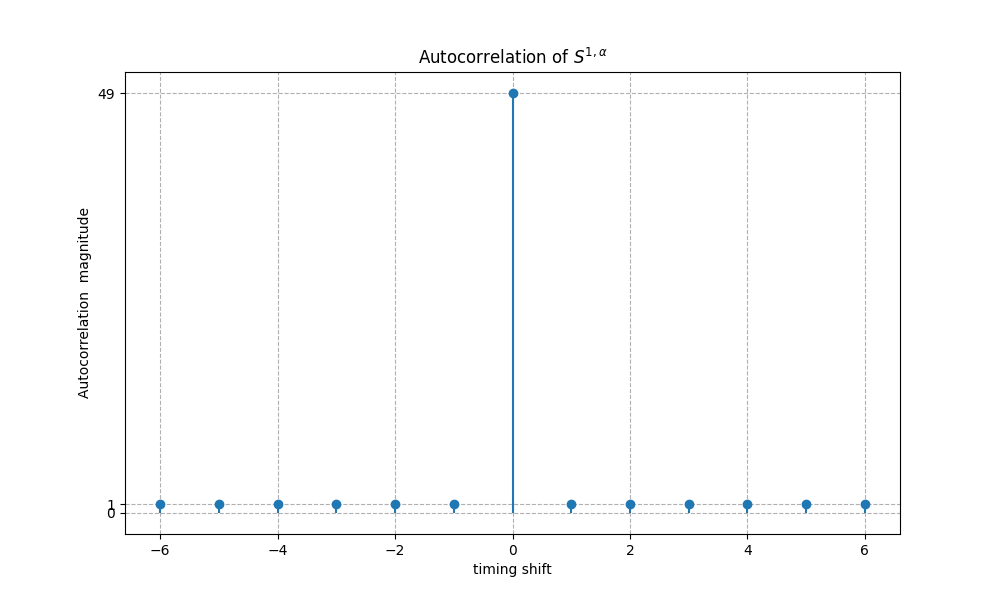}
\caption{The autocorrelation magnitude distribution of $\mathbf{S}^{1,\alpha}$ in Example \ref{example4}}
\label{fig13}
\end{figure*}
\begin{figure*}[htbp]
\centering
\includegraphics[width=0.6\columnwidth,height=0.4\linewidth]{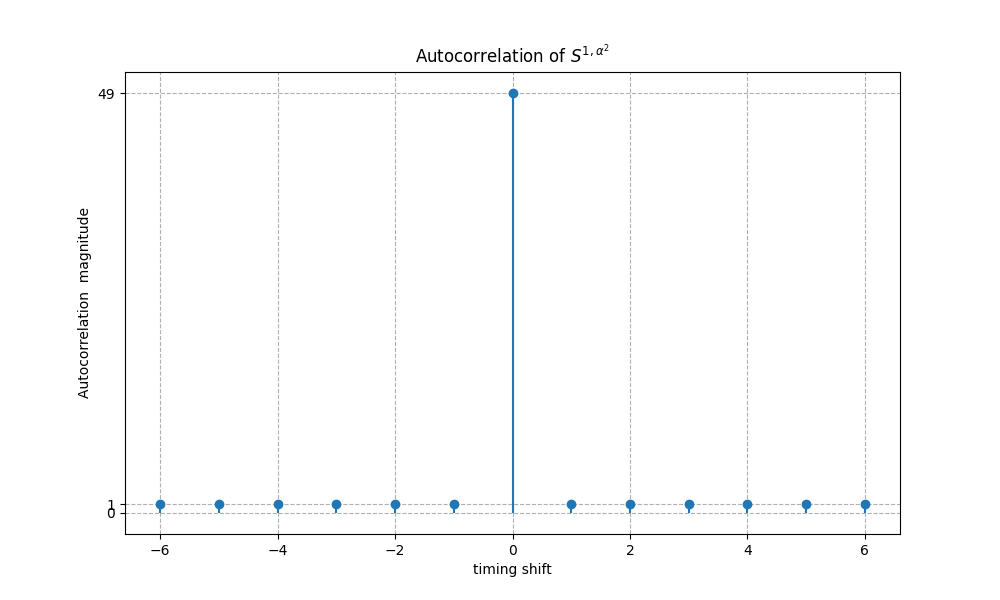}
\caption{The autocorrelation magnitude distribution  of $\mathbf{S}^{1,\alpha^2}$ in Example \ref{example4}}
\label{fig14}
\end{figure*}
\begin{figure*}[htbp]
\centering
\includegraphics[width=0.6\columnwidth,height=0.4\linewidth]{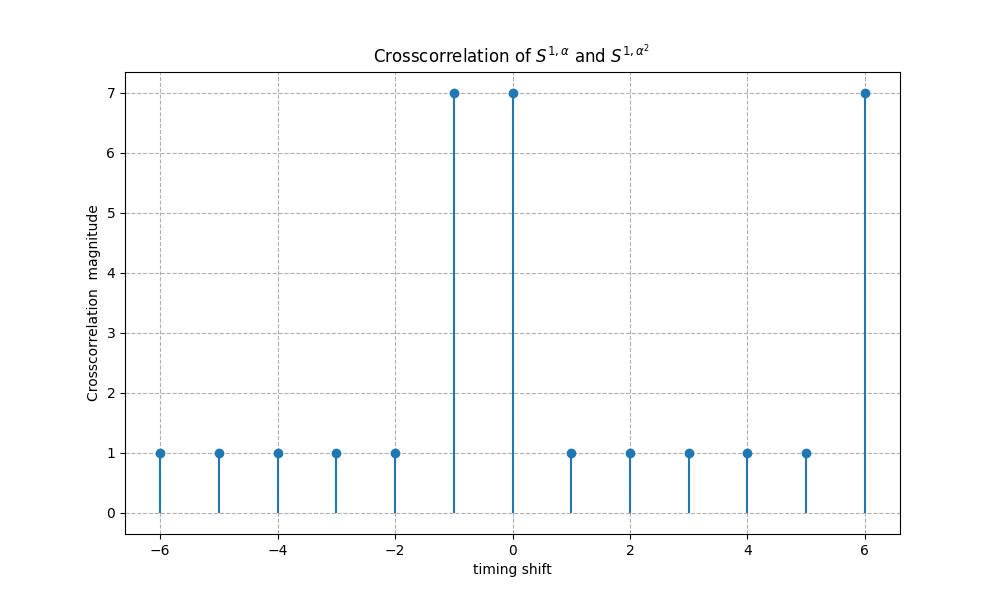}
\caption{The correlation magnitude distribution of $\mathbf{S}^{1,\alpha}$ and $\mathbf{S}^{1,\alpha^2}$ in Example \ref{example4}}
\label{fig15}
\end{figure*}
\end{example}

\section{Concluding remarks}
In this paper, we presented five new families of asymptotically optimal or near-optimal periodic QCSSs with large set sizes from some special polynomials and Gaussian sums.
Specially, we constructed the first two families of periodic QCSSs with set size $\Theta(K^3)$ and flock size $K$. 
Most of the QCSSs constructed in this paper have very small alphabet size. 
The  parameters of the QCSSs in this paper improve those of known ones in \cite{LY2, XLC}. 
From the figures in this paper, some complementary sequences have zero nontrivial autocorrelation magnitude.
These advantages indicate that our QCSSs have potential applications in MC-CDMA communication systems.

Note that the correlation lower bound in (\ref{wlech}) may not be tight if $M$ is too larger than $K$.
It is interesting to improve this bound. The reader is invited to study this interesting topic. 

\section*{Declarations}
\textbf{Conflict of interest} The authors declare that they have no conflicts of interest relevant to the content of this article.

\end{document}